\documentclass[12pt,technote, onecolumn, draft]{IEEEtran}
\usepackage{latexsym}
\usepackage{mathrsfs}
\usepackage{multirow}
\usepackage{amsfonts,amssymb,amsmath,amsthm,bm}
\usepackage{color}
\usepackage{cite}
\newtheorem{theorem}{Theorem}
\newtheorem{example}{Example}
\newtheorem{definition}{Definition}
\newtheorem{remark}{Remark}
\newtheorem{corollary}{Corollary}
\newtheorem{lemma}{Lemma}
\newtheorem{proposition}{Proposition}

\begin{document}

\title{Bent Partitions, Vectorial Dual-Bent Functions and Partial Difference Sets$^{\dag}$}
\author{Jiaxin Wang, Fang-Wei Fu, Yadi Wei
\IEEEcompsocitemizethanks{\IEEEcompsocthanksitem Jiaxin Wang, Fang-Wei Fu and Yadi Wei are with Chern Institute of Mathematics and LPMC, Nankai University, Tianjin 300071, China, Emails: wjiaxin@mail.nankai.edu.cn, fwfu@nankai.edu.cn, wydecho@mail.nankai.edu.cn.
}
\thanks{$^\dag$This research is supported by the National Key Research and Development Program of China (Grant Nos. 2018YFA0704703 and 2022YFA1005001), the National Natural Science Foundation of China (Grant Nos. 12141108, 61971243, 12226336), the Natural Science Foundation of Tianjin (20JCZDJC00610), the Fundamental Research Funds for the Central Universities of China (Nankai University), and the Nankai Zhide Foundation.}
}

\maketitle

\begin{abstract}
   Bent partitions of $V_{n}^{(p)}$ are quite powerful in constructing bent functions, vectorial bent functions and generalized bent functions, where $V_{n}^{(p)}$ is an $n$-dimensional vector space over $\mathbb{F}_{p}$, $n$ is an even positive integer and $p$ is a prime. It is known that partial spreads is a class of bent partitions. In \cite{AM2022Be,MP2021Be}, two classes of bent partitions whose forms are similar to partial spreads were presented. In \cite{AKM2022Ge}, more bent partitions $\Gamma_{1}, \Gamma_{2}, \Gamma_{1}^{\bullet}, \Gamma_{2}^{\bullet}, \Theta_{1}, \Theta_{2}$ were presented from (pre)semifields, including the bent partitions given in \cite{AM2022Be,MP2021Be}. In this paper, we investigate the relations between bent partitions and vectorial dual-bent functions. For any prime $p$, we show that one can generate certain bent partitions (called bent partitions satisfying Condition $\mathcal{C}$) from certain vectorial dual-bent functions (called vectorial dual-bent functions satisfying Condition A). In particular, when $p$ is an odd prime, we show that bent partitions satisfying Condition $\mathcal{C}$ one-to-one correspond to vectorial dual-bent functions satisfying Condition A. We give an alternative proof that $\Gamma_{1}, \Gamma_{2}, \Gamma_{1}^{\bullet}, \Gamma_{2}^{\bullet}, \Theta_{1}, \Theta_{2}$ are bent partitions in terms of vectorial dual-bent functions. We present a secondary construction of vectorial dual-bent functions, which can be used to generate more bent partitions. We show that any ternary weakly regular bent function $f: V_{n}^{(3)}\rightarrow \mathbb{F}_{3}$ ($n$ even) of $2$-form can generate a bent partition. When such $f$ is weakly regular but not regular, the generated bent partition by $f$ is not coming from a normal bent partition, which answers an open problem proposed in \cite{AM2022Be}. We give a sufficient condition on constructing partial difference sets from bent partitions, and when $p$ is an odd prime, we provide a characterization of bent partitions satisfying Condition $\mathcal{C}$ in terms of partial difference sets.
\end{abstract}


\section{Introduction}
\label{sec:1}
Boolean bent functions were introduced by Rothaus \cite{Rothaus1976On} and were generalized to $p$-ary bent functions by Kumar, Scholtz and Welch \cite{KSW1985Ge}, where $p$ is an arbitrary prime. Due to applications of $p$-ary bent functions in cryptography, coding theory, sequence and combinatorics, they have been extensively studied. We refer to surveys \cite{CS2016Fo,Meidl2022A} and a book \cite{Mesnager2016Be} on $p$-ary bent functions and their generalizations such as vectorial bent functions and generalized bent functions.

In \cite{CMP2018Ve}, \c{C}e\c{s}melio\u{g}lu \emph{et al.} introduced vectorial dual-bent functions, which is a special class of vectorial bent functions. In \cite{CMP2021Ve,CM2018Be,WF2022Ne}, vectorial dual-bent functions were used to construct partial difference sets. In particular, Wang and Fu \cite{WF2022Ne} showed that for certain vectorial dual-bent functions $F: V_{n}^{(p)}\rightarrow V_{s}^{(p)}$ (where $V_{n}^{(p)}$ is an $n$-dimensional vector space over the prime field $\mathbb{F}_{p}$), the preimage set of any subset of $V_{s}^{(p)}$ for $F$ forms a partial difference set.

Very recently, bent partitions of $V_{n}^{(p)}$ were introduced \cite{AM2022Be,MP2021Be}, which are quite powerful in constructing bent functions, vectorial bent functions and generalized bent functions. The well-known partial spreads is a class of bent partitions. In \cite{MP2021Be}, Meidl and Pirsic for the first time presented two classes of bent partitions for $p=2$ different from partial spreads. In \cite{AM2022Be}, Anbar and Meidl generalized the contributions in \cite{MP2021Be} to the case of $p$ being odd and gave the corresponding two classes of bent partitions for odd $p$. In \cite{AKM2022Ge}, Anbar, Kalayc{\i} and Meidl presented more bent partitions $\Gamma_{1}, \Gamma_{2}, \Gamma_{1}^{\bullet}, \Gamma_{2}^{\bullet}, \Theta_{1}, \Theta_{2}$ from (pre)semifields, including the bent partitions given in \cite{AM2022Be,MP2021Be}. In \cite{AKM2022Be}, Anbar, Kalayc{\i} and Meidl showed that any union of elements in the bent partition given in \cite{AM2022Be,MP2021Be} forms a partial difference set. In terms of constructing partial difference sets, certain vectorial dual-bent functions and certain bent partitions seem to play the same role. Therefore, it is interesting to investigate the relations between vectorial dual-bent functions and bent partitions. In this paper, we show that by using certain vectorial dual-bent functions (called vectorial dual-bent functions satisfying Condition A), we can construct bent partitions of $V_{n}^{(p)}$ with certain properties (called bent partitions satisfying Condition $\mathcal{C}$) for any prime $p$. Particularly, when $p$ is an odd prime, we prove that bent partitions of $V_{n}^{(p)}$ with Condition $\mathcal{C}$ one-to-one correspond to vectorial dual-bent functions satisfying Condition A. In terms of vectorial dual-bent functions, we provide an alternative proof that $\Gamma_{1}, \Gamma_{2}, \Gamma_{1}^{\bullet}, \Gamma_{2}^{\bullet}, \Theta_{1}, \Theta_{2}$ given in \cite{AKM2022Ge} are bent partitions. We provide a secondary construction of vectorial dual-bent functions, which can be used to generate more bent partitions. We prove that any ternary weakly regular bent function $f: V_{n}^{(3)}\rightarrow \mathbb{F}_{3}$ ($n$ even) of $2$-form can generate a bent partition. In the special case that $f$ is weakly regular but not regular, the generated bent partition by $f$ is not coming from a normal bent partition, which answers an open problem proposed in \cite{AM2022Be}. By using vectorial dual-bent functions as the link between bent partitions and partial difference sets, we give a sufficient condition on constructing partial difference sets from bent partitions. When $p$ is an odd prime, we provide a characterization of bent partitions satisfying Condition $\mathcal{C}$ in terms of partial difference sets.

The rest of the paper is organized as follows. In Section II, we state some needed results on vectorial dual-bent functions and bent partitions. In Section III, we present relations between certain bent partitions and certain vectorial dual-bent functions. In Section IV, we give a secondary construction of vectorial dual-bent functions, which can be used to generate more bent partitions. In Section V, we present relations between certain bent partitions and certain partial difference sets. In Section VI, we make a conclusion.

\section{Preliminaries}
\label{sec:2}
In this section, we state some basic results on vectorial dual-bent functions and bent partitions. First, we fix some notations used throughout this paper.
\begin{itemize}
  \item $p$ is a prime.
  \item $\zeta_{p}=e^{\frac{2 \pi \sqrt{-1}}{p}}$ is a complex primitive $p$-th root of unity. Note that $\zeta_{2}=-1$.
  \item $\mathbb{F}_{p^n}$ is the finite field with $p^n$ elements.
  \item $\mathbb{F}_{p}^{n}$ is the vector space of the $n$-tuples over $\mathbb{F}_{p}$.
  \item $V_{n}^{(p)}$ is an $n$-dimensional vector space over $\mathbb{F}_{p}$.
  \item $\langle \cdot \rangle_{n}$ denotes a (non-degenerate) inner product of $V_{n}^{(p)}$. In this paper, when $V_{n}^{(p)}=\mathbb{F}_{p^n}$, let $\langle a, b\rangle_{n}=Tr_{1}^{n}(ab)$, where $a, b \in \mathbb{F}_{p^n}$, $Tr_{k}^{n}(\cdot)$ denotes the trace function from $\mathbb{F}_{p^n}$ to $\mathbb{F}_{p^k}$, $k \mid n$; when $V_{n}^{(p)}=\mathbb{F}_{p}^{n}$, let $\langle a, b\rangle_{n}=a \cdot b=\sum_{i=1}^{n}a_{i}b_{i}$, where $a=(a_{1}, \dots, a_{n}), b=(b_{1}, \dots, b_{n})\in \mathbb{F}_{p}^{n}$; when $V_{n}^{(p)}=V_{n_{1}}^{(p)}\times \dots \times V_{n_{m}}^{(p)} (n=\sum_{i=1}^{m}n_{i})$, let $\langle a, b\rangle_{n}=\sum_{i=1}^{m}\langle a_{i}, b_{i}\rangle_{n_{i}}$, where $a=(a_{1}, \dots, a_{m}), b=(b_{1}, \dots, b_{m})\in V_{n}^{(p)}$.
  \item For any set $A\subseteq V_{n}^{(p)}$ and $u \in V_{n}^{(p)}$, let $\chi_{u}(A)=\sum_{x \in A}\chi_{u}(x)$, where $\chi_{u}$ denotes the character $\chi_{u}(x)=\zeta_{p}^{\langle u, x\rangle_{n}}$.
\end{itemize}
\subsection{Vectorial dual-bent functions}
A function $F: V_{n}^{(p)}\rightarrow V_{s}^{(p)}$ is called a \textit{vectorial $p$-ary function}, or simply \textit{$p$-ary function} when $s=1$. The \textit{Walsh transform} of a $p$-ary function $f: V_{n}^{(p)}\rightarrow \mathbb{F}_{p}$ is the complex valued function defined by
\begin{equation}\label{1}
  W_{f}(a)=\sum_{x \in V_{n}^{(p)}}\zeta_{p}^{f(x)-\langle a, x\rangle_{n}}, a \in V_{n}^{(p)}.
\end{equation}

A $p$-ary function $f: V_{n}^{(p)}\rightarrow \mathbb{F}_{p}$ is called \textit{bent} if $|W_{f}(a)|=p^{\frac{n}{2}}$ for all $a \in V_{n}^{(p)}$. Note that when $f$ is a Boolean bent function, that is, $p=2$, then $n$ must be even since in this case, $W_{f}$ is an integer valued function. A vectorial $p$-ary function $F: V_{n}^{(p)}\rightarrow V_{s}^{(p)}$ is called \textit{vectorial bent} if all \textit{component functions} $F_{c}: V_{n}^{(p)}\rightarrow \mathbb{F}_{p}, c \in V_{s}^{(p)}\backslash \{0\}$ defined as $F_{c}(x)=\langle c, F(x)\rangle_{s}$ are bent. It is known that if $F: V_{n}^{(p)}\rightarrow V_{s}^{(p)}$ is vectorial bent, then $s\leq \frac{n}{2}$ if $p=2$, and $s\leq n$ if $p$ is an odd prime. If $f: V_{n}^{(p)}\rightarrow \mathbb{F}_{p}$ is bent, then so are $cf, c \in \mathbb{F}_{p}^{*}$, that is, any $p$-ary bent function is vectorial bent. For $F: V_{n}^{(p)}\rightarrow V_{s}^{(p)}$, the vectorial bentness of $F$ is independent of the inner products of $V_{n}^{(p)}$ and $V_{s}^{(p)}$. The Walsh transform of a $p$-ary bent function $f: V_{n}^{(p)} \rightarrow \mathbb{F}_{p}$ satisfies that for any $a \in V_{n}^{(p)}$, when $p=2$, we have
\begin{equation}\label{2}
  W_{f}(a)=2^{\frac{n}{2}}(-1)^{f^{*}(a)},
\end{equation}
and when $p$ is an odd prime, we have
\begin{equation}\label{3}
  W_{f}(a)=\left\{\begin{split}
                     \pm p^{\frac{n}{2}}\zeta_{p}^{f^{*}(a)} & \ \ \text{if} \ p \equiv 1 \ (mod \ 4) \ \text{or} \ n \ \text{is even},\\
                     \pm \sqrt{-1} p^{\frac{n}{2}} \zeta_{p}^{f^{*}(a)} & \ \ \text{if} \ p \equiv 3 \ (mod \ 4) \ \text{and} \ n \ \text{is odd},
                  \end{split}\right.
\end{equation}
where $f^{*}$ is a function from $V_{n}^{(p)}$ to $\mathbb{F}_{p}$, called the \textit{dual} of $f$. A $p$-ary bent function $f: V_{n}^{(p)} \rightarrow \mathbb{F}_{p}$ is called \textit{weakly regular} if $W_{f}(a)=\varepsilon_{f}p^{\frac{n}{2}}\zeta_{p}^{f^{*}(a)}$, where $\varepsilon_{f}$ is a constant independent of $a$, otherwise $f$ is called \textit{non-weakly regular}. In particular, if $\varepsilon_{f}=1$, $f$ is called \textit{regular}. The (non-)weakly regularity of $f$ is independent of the inner product of $V_{n}^{(p)}$ and if $f$ is weakly regular, $\varepsilon_{f}$ is independent of the inner product of $V_{n}^{(p)}$. By (2), all Boolean bent functions are regular. If $f$ is a $p$-ary weakly regular bent function, then the dual $f^{*}$ of $f$ is also weakly regular bent with $(f^{*})^{*}(x)=f(-x)$ (see \cite{CMP2013On}).

In 2018, \c{C}e\c{s}melio\u{g}lu \emph{et al.} \cite{CMP2018Ve} introduced vectorial dual-bent functions.
\begin{definition}\label{1}
A vectorial $p$-ary bent function $F: V_{n}^{(p)}\rightarrow V_{s}^{(p)}$ is called \textit{vectorial dual-bent} if there exists a vectorial bent function $G: V_{n}^{(p)}\rightarrow V_{s}^{(p)}$ such that $(F_{c})^{*}=G_{\sigma(c)}$ for any $c \in V_{s}^{(p)}\backslash \{0\}$, where $(F_{c})^{*}$ is the dual of the component function $\langle c, F(x)\rangle_{s}$ and $\sigma$ is some permutation over $V_{s}^{(p)}\backslash \{0\}$. The vectorial bent function $G$ is called a \textit{vectorial dual} of $F$ and denoted by $F^{*}$.
\end{definition}
It is known in \cite{CMP2018Ve} that the property of being vectorial dual-bent is independent of the inner products of $V_{n}^{(p)}$ and $V_{s}^{(p)}$. Note that for a vectorial dual-bent function, its vectorial dual is not unique since being vectorial bent and vectorial dual-bent for a function is a property of the vector space consisting of all component functions (see Remark 1 of \cite{CMP2018Ve}). For example, if a $p$-ary function $f$ (seen as a vectorial function into $V_{1}^{(p)}$, $p$ odd) is vectorial dual-bent under any fixed inner product, then its dual $f^{*}$ is unique, but its vectorial dual is not unique since for any $c \in \mathbb{F}_{p}^{*}$, $cf^{*}$ is a vectorial dual of $f$. A $p$-ary function $f: V_{n}^{(p)}\rightarrow \mathbb{F}_{p}$ is called an \textit{$l$-form} if $f(ax)=a^{l}f(x)$ for any $a \in \mathbb{F}_{p}^{*}$ and $x \in V_{n}^{(p)}$, where $1\leq l\leq p-1$ is an integer. By the results in \cite{CM2018Be,WF2022Ne}, we have the following proposition.

\begin{proposition}[\cite{CM2018Be,WF2022Ne}]\label{1}
A $p$-ary function $f$ with $f(0)=0$ is a weakly regular vectorial dual-bent function if and only if $f$ is a weakly regular bent function of $l$-form with $gcd(l-1,p-1)=1$. In particular, a $p$-ary function $f$ is a weakly regular vectorial dual-bent function with $(cf)^{*}=cf^{*}$ for any $c \in \mathbb{F}_{p}^{*}$ if and only if $f$ is a weakly regular bent function of $(p-1)$-form.
\end{proposition}

In the rest of this subsection, we recall an important class of $p$-ary bent functions, called Maiorana-McFarland bent functions.
\begin{itemize}
  \item Let $f: \mathbb{F}_{p^n}\times \mathbb{F}_{p^n}\rightarrow \mathbb{F}_{p}$ be defined as
  \begin{equation*}
    f(x, y)=Tr_{1}^{n}(\alpha x \pi(y))+g(y),
  \end{equation*}
  where $\alpha \in \mathbb{F}_{p^n}^{*}$, $\pi$ is a permutation over $\mathbb{F}_{p^n}$ and $g: \mathbb{F}_{p^n}\rightarrow \mathbb{F}_{p}$ is an arbitrary function. Then $f$ is bent and is called a Maiorana-McFarland bent function. The dual $f^{*}$ of $f$ is
  \begin{equation}\label{4}
    f^{*}(x, y)=Tr_{1}^{n}(-\pi^{-1}(\alpha^{-1}x)y)+g(\pi^{-1}(\alpha^{-1}x)).
  \end{equation}
   All Maiorana-McFarland bent functions are regular (see \cite{KSW1985Ge}).
\end{itemize}

\subsection{Bent partitions}
Very recently, the concept of bent partitions of $V_{n}^{(p)}$ were introduced \cite{AM2022Be,MP2021Be}.

\begin{definition}\label{2}
Let $n$ be an even positive integer, $K$ be a positive integer divisible by $p$.
\begin{itemize}
  \item Let $\Gamma=\{A_{1}, \dots, A_{K}\}$ be a partition of $V_{n}^{(p)}$. Assume that every function $f$ for which every $i\in \mathbb{F}_{p}$ has exactly $\frac{K}{p}$ of sets $A_{j}$ in $\Gamma$ in its preimage, is a $p$-ary bent function. Then $\Gamma$ is called a \textit{bent partition} of $V_{n}^{(p)}$ of depth $K$ and every such bent function $f$ is called a \textit{bent function constructed from bent partition $\Gamma$}.
  \item Let $\Gamma=\{U, A_{1}, \dots, A_{K}\}$ be a partition of $V_{n}^{(p)}$. Assume that every function $f$ with the following properties is bent:

  (1) Every $c \in \mathbb{F}_{p}$ has exactly $\frac{K}{p}$ of the sets $A_{1}, \dots, A_{K}$ in its preimage set;

  (2) $f(x)=c_{0}$ for all $x \in U$ and some fixed $c_{0} \in \mathbb{F}_{p}$.

  Then we call $\Gamma$ a \textit{normal bent partition} of $V_{n}^{(p)}$ of depth $K$.
\end{itemize}
\end{definition}

Bent partitions are very powerful in constructing bent functions, vectorial bent function and generalized bent functions. In this paper, we focus on the relations between bent partitions and vectorial bent functions.

\begin{proposition}[\cite{AM2022Be}]\label{2}
Let $\Gamma=\{A_{1}, \dots, A_{p^s}\}$ be a bent partition of $V_{n}^{(p)}$. Then every function $F: V_{n}^{(p)}\rightarrow V_{s}^{(p)}$ such that every element $i \in V_{s}^{(p)}$ has the elements of exactly one of the sets $A_{j}, 1\leq j \leq p^s$, in its preimage, is a vectorial bent function.
\end{proposition}

It is known that partial spreads is a class of bent partitions (for instance see Section 2 of \cite{AM2022Be}). In \cite{AM2022Be,MP2021Be}, two classes of explicit bent partitions different from partial spreads were presented. In \cite{AKM2022Ge}, bent partitions $\Gamma_{1}, \Gamma_{2}, \Gamma_{1}^{\bullet}, \Gamma_{2}^{\bullet}, \Theta_{1}, \Theta_{2}$ were presented from certain (pre)semifields, including the bent partitions given in \cite{AM2022Be,MP2021Be}. We will recall bent partitions $\Gamma_{1}, \Gamma_{2}, \Gamma_{1}^{\bullet}, \Gamma_{2}^{\bullet}, \Theta_{1}, \Theta_{2}$ given in \cite{AKM2022Ge}. First, we need to introduce some basic knowledge on (pre)semifields.

\begin{definition}\label{3}
Let $\circ$ be a binary operation defined on $(V_{n}^{(p)}, +)$ such that

(i) $x \circ y=0$ implies $x=0$ or $y=0$,

(ii) $(x+y)\circ z=(x\circ z)+(y \circ z), (z\circ (x+y)=(z\circ x)+(z\circ y)$, respectively),

\noindent for all $x, y, z \in V_{n}^{(p)}$. Then $(V_{n}^{(p)}, +, \circ)$ is called a \textit{right} (\textit{left}, respectively) \textit{prequasifield}. If $(V_{n}^{(p)}, +, \circ)$ is a right and a left prequasifield, then it is called a \textit{presemifield}. If $(V_{n}^{(p)}, +, \circ)$ is a presemifield for which there is an element $e\neq 0$ such that $e\circ x=x\circ e=x$ for all $x \in V_{n}^{(p)}$, then it is called a \textit{semifield}.
\end{definition}

Let $P=(\mathbb{F}_{p^n}, +, \circ)$ be a presemifield. Then one can obtain presemifields $P^{\bullet}=(\mathbb{F}_{p^n}, +, \bullet)$ and $P^{\star}=(\mathbb{F}_{p^n}, +, \star)$ from $P$, where $\bullet$ and $\star$ are given by
\begin{equation*}
  x \bullet y=y\circ x \text{ for all }x, y \in \mathbb{F}_{p^n},
\end{equation*}
and
\begin{equation*}
  Tr_{1}^{n}(z(x\circ y))=Tr_{1}^{n}(x(z\star y)) \text{ for all }x, y, z \in \mathbb{F}_{p^n},
\end{equation*}
respectively. The presemifield $P^{\star}$ is called the \textit{dual} of $P$. Let $s$ be a positive divisor of $n$. If $x\circ (cy)=c(x\circ y)$ holds for any $x, y \in \mathbb{F}_{p^n}, c \in \mathbb{F}_{p^s}$, then $P$ is called \textit{right $\mathbb{F}_{p^s}$-linear}. Each presemifield $P=(\mathbb{F}_{p^n}, +, \circ)$ can induce a semifield $P'=(\mathbb{F}_{p^n}, +, \ast)$ via the following transformation: choose any $\alpha \in \mathbb{F}_{p^n}^{*}$ and give $\ast$ by
\begin{equation*}
  (x\circ \alpha)\ast (\alpha\circ y)=x\circ y \text{ for all }x, y \in \mathbb{F}_{p^n}.
\end{equation*}
By Lemma 2 of \cite{AKM2022Ge}, if $P$ is right $\mathbb{F}_{p^s}$-linear, then $P'$ is also right $\mathbb{F}_{p^s}$-linear. The finite field $\mathbb{F}_{p^n}$ is a right $\mathbb{F}_{p^s}$-linear semifield (that is, $\circ$ is the field multiplication). For more right $\mathbb{F}_{p^s}$-linear (pre)semifields, see Section 3 of \cite{AKM2022Ge}.

Now we recall bent partitions $\Gamma_{1}, \Gamma_{2}, \Gamma_{1}^{\bullet}, \Gamma_{2}^{\bullet}, \Theta_{1}, \Theta_{2}$ given in \cite{AKM2022Ge}.

\begin{itemize}
  \item Let $n, s$ be positive integers satisfying $s \mid n$ and $gcd(p^n-1,p^s+p-1)=1$. Set $u=p^s+p-1$, and let $d$ be an integer with $du\equiv 1 \ mod \ (p^n-1)$. Let $P=(\mathbb{F}_{p^n}, +, \circ)$ be a (pre)semifield such that its dual $P^{\star}=(\mathbb{F}_{p^n}, +, \star)$ is right $\mathbb{F}_{p^s}$-linear. For given $x \in \mathbb{F}_{p^n}$, if $x=0$, then let $\eta_{x}=0$, and if $x\neq 0$, then let $\eta_{x}$ be given by $x\star \eta_{x}^{-1}=1$.
  \item Define
      \begin{equation*}
        U_{t}=\{(x, t\circ x^{u}): x \in \mathbb{F}_{p^n}^{*}\} \text{ if } t \in \mathbb{F}_{p^n}, \text{ and }U= \{(0, y): y \in \mathbb{F}_{p^n}\}.
      \end{equation*}
      Let $i_{0} \in \mathbb{F}_{p^s}$ be an arbitrary element. Define
      \begin{equation}\label{5}
       \Gamma_{1}=\{A_{i}, i \in \mathbb{F}_{p^s}\},
      \end{equation}
      where $A_{i}=\bigcup_{t \in \mathbb{F}_{p^n}: Tr_{s}^{n}(t)=i}U_{t}$ if $i\neq i_{0}$, $A_{i_{0}}=\bigcup_{t \in \mathbb{F}_{p^n}: Tr_{s}^{n}(t)=i_{0}}U_{t} \bigcup U$.
  \item Define
      \begin{equation*}
        U_{t}^{\bullet}=\{(x, x^{u}\circ t): x \in \mathbb{F}_{p^n}^{*}\} \text{ if } t \in \mathbb{F}_{p^n}, \text{ and }U= \{(0, y): y \in \mathbb{F}_{p^n}\}.
      \end{equation*}
      Let $i_{0} \in \mathbb{F}_{p^s}$ be an arbitrary element. Define
      \begin{equation}\label{6}
       \Gamma_{1}^{\bullet}=\{A_{i}^{\bullet}, i \in \mathbb{F}_{p^s}\},
      \end{equation}
      where $A_{i}^{\bullet}=\bigcup_{t \in \mathbb{F}_{p^n}: Tr_{s}^{n}(t)=i}U_{t}^{\bullet}$ if $i\neq i_{0}$, $A_{i_{0}}^{\bullet}=\bigcup_{t \in \mathbb{F}_{p^n}: Tr_{s}^{n}(t)=i_{0}}U_{t}^{\bullet} \bigcup U$.
  \item Define
      \begin{equation*}
        V_{t}=\{(t\circ x^{d}, x): x \in \mathbb{F}_{p^n}^{*}\} \text{ if } t \in \mathbb{F}_{p^n}, \text{ and }V=\{(x, 0):  x \in \mathbb{F}_{p^n}\}.
      \end{equation*}
      Let $i_{0} \in \mathbb{F}_{p^s}$ be an arbitrary element. Define
      \begin{equation}\label{7}
       \Gamma_{2}=\{B_{i}, i \in \mathbb{F}_{p^s}\},
      \end{equation}
     where $B_{i}=\bigcup_{t \in \mathbb{F}_{p^n}: Tr_{s}^{n}(t)=i}V_{i}$ if $i\neq i_{0}$, $B_{i_{0}}=\bigcup_{t \in \mathbb{F}_{p^n}: Tr_{s}^{n}(t)=i_{0}}V_{i}\bigcup V$.
  \item Define
      \begin{equation*}
        V_{t}^{\bullet}=\{(x^{d}\circ t, x): x \in \mathbb{F}_{p^n}^{*}\} \text{ if } t \in \mathbb{F}_{p^n}, \text{ and }V=\{(x, 0):  x \in \mathbb{F}_{p^n}\}.
      \end{equation*}
      Let $i_{0} \in \mathbb{F}_{p^s}$ be an arbitrary element. Define
      \begin{equation}\label{8}
       \Gamma_{2}^{\bullet}=\{B_{i}^{\bullet}, i \in \mathbb{F}_{p^s}\},
      \end{equation}
     where $B_{i}^{\bullet}=\bigcup_{t \in \mathbb{F}_{p^n}: Tr_{s}^{n}(t)=i}V_{i}^{\bullet}$ if $i\neq i_{0}$, $B_{i_{0}}=\bigcup_{t \in \mathbb{F}_{p^n}: Tr_{s}^{n}(t)=i_{0}}V_{i}^{\bullet}\bigcup V$.
  \item Define
      \begin{equation*}
        X_{t}=\{(t\eta_{x}^{d}, x): x \in \mathbb{F}_{p^n}^{*}\} \text{ if } t \in \mathbb{F}_{p^n}, \text{ and }X=\{(x, 0): x \in \mathbb{F}_{p^n}\}.
      \end{equation*}
      Let $i_{0} \in \mathbb{F}_{p^s}$ be an arbitrary element. Define
      \begin{equation}\label{9}
       \Theta_{1}=\{S_{i}, i \in \mathbb{F}_{p^s}\},
      \end{equation}
      where $S_{i}=\bigcup_{t \in \mathbb{F}_{p^n}: Tr_{s}^{n}(t)=i}X_{t}$ if $i\neq i_{0}$, $S_{i_{0}}=\bigcup_{t \in \mathbb{F}_{p^n}: Tr_{s}^{n}(t)=i_{0}}X_{t} \bigcup X$.
  \item Define
      \begin{equation*}
        Y_{t}=\{(x, t\eta_{x}^{u}): x \in \mathbb{F}_{p^n}^{*}\} \text{ if } t \in \mathbb{F}_{p^n}, \text{ and }Y= \{(0, y): y \in \mathbb{F}_{p^n}\}.
      \end{equation*}
      Let $i_{0} \in \mathbb{F}_{p^s}$ be an arbitrary element. Define
      \begin{equation}\label{10}
       \Theta_{2}=\{T_{i}, i \in \mathbb{F}_{p^s}\},
      \end{equation}
      where $T_{i}=\bigcup_{t \in \mathbb{F}_{p^n}: Tr_{s}^{n}(t)=i}Y_{t}$ if $i\neq i_{0}$, $T_{i_{0}}=\bigcup_{t \in \mathbb{F}_{p^n}: Tr_{s}^{n}(t)=i_{0}}Y_{t} \bigcup Y$.
\end{itemize}

\begin{remark}\label{1}
In the finite field case, that is, $\circ$ and $\star$ are the field multiplication, then $\Gamma_{1}=\Gamma_{1}^{\bullet}=\Theta_{2}, \Gamma_{2}=\Gamma_{2}^{\bullet}=\Theta_{1}$, which reduces to the two classes bent partitions given in \cite{AM2022Be,MP2021Be}.
\end{remark}

\begin{remark}\label{2}
In fact, for the parameter $u$ in the bent partitions $\Gamma_{1}, \Gamma_{1}^{\bullet}, \Gamma_{2}, \Gamma_{2}^{\bullet}, \Theta_{1}, \Theta_{2}$, one can consider the more general form $u\equiv p^{j} \ mod \ (p^s-1)$ by the proofs in \cite{AKM2022Ge}.
\end{remark}

\section{Relations between certain bent partitions and certain vectorial dual-bent functions}
\label{sec:3}
Throughout this section, we consider bent partitions and vectorial dual-bent functions satisfying the following conditions, respectively.

\textbf{Condition $\mathcal{C}$}: Let $n$ be an even positive integer, $s$ be a positive integer with $s\leq \frac{n}{2}$. Let $\Gamma=\{A_{i}, i \in V_{s}^{(p)}\}$ be a bent partition of $V_{n}^{(p)}$ which satisfies that $\mathbb{F}_{p}^{*}A_{i}=A_{i}$ for all $i \in V_{s}^{(p)}$ and all bent functions $f$ constructed from $\Gamma$ are regular (that is, $\varepsilon_{f}=1$) or weakly regular but not regular (that is, $\varepsilon_{f}=-1$). We denote by $\varepsilon=\varepsilon_{f}$ for all bent functions $f$ constructed from $\Gamma$.

\textbf{Condition A}: Let $n$ be an even positive integer, $s$ be a positive integer with $s\leq \frac{n}{2}$. Let $F: V_{n}^{(p)}\rightarrow V_{s}^{(p)}$ be a vectorial dual-bent function with $(F_{c})^{*}=(F^{*})_{c}, c \in V_{s}^{(p)}\backslash \{0\}$ for a vectorial dual $F^{*}$ of $F$ and all component functions being regular or weakly regular but not regular, that is, $\varepsilon_{F_{c}}, c \in V_{s}^{(p)} \backslash \{0\}$ are all the same. We denote by $\varepsilon=\varepsilon_{F_{c}}$ for all $c \in V_{s}^{(p)}\backslash \{0\}$.

It is easy to see that the known bent partitions, including partial spreads and $\Gamma_{i}, \Gamma_{i}^{\bullet}, \Theta_{i}, i=1, 2$ defined by (5)-(10), all satisfy $\mathbb{F}_{p}^{*}A_{i}=A_{i}, i \in V_{s}^{(p)}$. By the results in \cite{Dillon1974El,LL2014Be,AKM2022Ge}, all bent functions constructed from partial spreads and $\Gamma_{i}, \Gamma_{i}^{\bullet}, \Theta_{i}, i=1, 2$ are regular. Thus, the known bent partitions all satisfy Condition $\mathcal{C}$ with $\varepsilon=1$. Moreover, when $p=2$, it is easy to see that Condition $\mathcal{C}$ is trivial for any bent partition of $V_{n}^{(2)}$ of depth powers of $2$. In this section, we present relations between bent partitions satisfying Condition $\mathcal{C}$ and vectorial dual-bent functions satisfying Condition A. First, we need a lemma.

\begin{lemma}\label{1}
Let $n$ be an even positive integer, $s$ be a positive integer with $s\leq \frac{n}{2}$, and $F: V_{n}^{(p)}\rightarrow V_{s}^{(p)}$. Then the following two statements are equivalent.

(1) $F$ is a vectorial dual-bent function satisfying Condition A.

(2) There exist pairwise disjoint sets $W_{i} \subseteq V_{n}^{(p)}, i \in V_{s}^{(p)}$ with $\bigcup_{i \in V_{s}^{(p)}}W_{i}=V_{n}^{(p)}$ and a constant $\varepsilon \in \{\pm1\}$ ($\varepsilon=1$ if $p=2$) such that for any nonempty set $I\subseteq V_{s}^{(p)}$,
\begin{equation}\label{11}
  \chi_{u}(D_{F,I})=p^{n-s}\delta_{\{0\}}(u)|I|+\varepsilon p^{\frac{n}{2}-s}(p^{s}\delta_{W_{I}}(u)-|I|), u \in V_{n}^{(p)},
\end{equation}
where $D_{F, I}=\{x \in V_{n}^{(p)}: F(x) \in I\}$, $W_{I}=\bigcup_{i \in I}W_{i}$, and for any set $S$, $\delta_{S}$ denotes the indicator function of $S$.
\end{lemma}

\begin{proof}
By Proposition 3 of \cite{WF2022Ne} (Note that although Proposition 3 of \cite{WF2022Ne} only considers the case of $p$ being odd, $p=2$ also holds), for any $u \in V_{n}^{(p)}, i \in V_{s}^{(p)}$ we have
\begin{equation}\label{12}
     \chi_{u}(D_{F, i})=p^{n-s}\delta_{\{0\}}(u)+p^{-s}\sum_{c \in V_{s}^{(p)}\backslash \{0\}}W_{F_{c}}(-u)\zeta_{p}^{-\langle c, i\rangle_{s}},
\end{equation}
where $D_{F, i}=\{x \in V_{n}^{(p)}: F(x)=i\}$.

(1) $\Rightarrow$ (2): If $F$ is a vectorial dual-bent function satisfying Condition A (Note that if $p=2$, then $\varepsilon=1$ since all Boolean bent functions are regular), then
\begin{equation}\label{13}
\begin{split}
\chi_{u}(D_{F, i})&=p^{n-s}\delta_{\{0\}}(u)+\varepsilon p^{\frac{n}{2}-s}\sum_{c \in V_{s}^{(p)}\backslash \{0\}}\zeta_{p}^{(F_{c})^{*}(-u)-\langle c, i\rangle_{s}}\\
&=p^{n-s}\delta_{\{0\}}(u)+\varepsilon p^{\frac{n}{2}-s}\sum_{c \in V_{s}^{(p)}\backslash \{0\}}\zeta_{p}^{(F^{*})_{c}(-u)-\langle c, i\rangle_{s}} \\
&=p^{n-s}\delta_{\{0\}}(u)+\varepsilon p^{\frac{n}{2}-s}\sum_{c \in V_{s}^{(p)}\backslash \{0\}}\zeta_{p}^{\langle c, F^{*}(-u)-i\rangle_{s}} \\ &=p^{n-s}\delta_{\{0\}}(u)+\varepsilon p^{\frac{n}{2}-s}(p^{s}\delta_{\{0\}}(F^{*}(-u)-i)-1).\\
\end{split}
\end{equation}
Define $W_{i}=\{x \in V_{n}^{(p)}: F^{*}(-x)=i\}, i \in V_{s}^{(p)}$. Then $W_{i}\bigcap W_{j}=\emptyset$ for any $i\neq j$ and $\bigcup_{i \in V_{s}^{(p)}}W_{i}=V_{n}^{(p)}$. By (13), for any nonempty set $I\subseteq V_{s}^{(p)}$ and $u \in V_{n}^{(p)}$ we have
\begin{equation*}
\begin{split}
 \chi_{u}(D_{F, I})&=\sum_{i \in I}\chi_{u}(D_{F, i})\\
 &=\sum_{i \in I}p^{n-s}\delta_{\{0\}}(u)+\varepsilon p^{\frac{n}{2}-s}(p^{s}\delta_{W_{i}}(u)-1)\\
 &=p^{n-s}\delta_{\{0\}}(u)|I|+\varepsilon p^{\frac{n}{2}-s}(p^{s}\delta_{W_{I}}(u)-|I|).\\
 \end{split}
\end{equation*}

(2) $\Rightarrow$ (1): By the assumption on $W_{i}, i \in V_{s}^{(p)}$, we have that for any $x \in V_{n}^{(p)}$, there exists a unique $i \in V_{s}^{(p)}$ such that $x \in W_{i}$. Define $G: V_{n}^{(p)}\rightarrow V_{s}^{(p)}$ by
\begin{equation*}
  G(x)=i \text{ if } -x\in W_{i}.
\end{equation*}
By the definition of $G$, for any $u \in V_{n}^{(p)}, i \in V_{s}^{(p)}$ we have
\begin{equation}\label{14}
  \begin{split}
   \chi_{u}(D_{F, i})=p^{n-s}\delta_{\{0\}}(u)+\varepsilon p^{\frac{n}{2}-s}(p^s\delta_{\{0\}}(G(-u)-i)-1).
  \end{split}
\end{equation}
For any $c \in V_{s}^{(p)}\backslash \{0\}$,
\begin{equation*}
  \begin{split}
  W_{F_{c}}(-u)&=\sum_{x \in V_{n}^{(p)}}\zeta_{p}^{\langle c, F(x)\rangle_{s}+\langle u, x\rangle_{n}}\\
 &=\sum_{i \in V_{s}^{(p)}}\sum_{x \in V_{n}^{(p)}: F(x)=i}\zeta_{p}^{\langle c, i\rangle_{s}+\langle u, x\rangle_{n}}\\
 \end{split}
 \end{equation*}
 \begin{equation}\label{15}
 \begin{split}
 &=\sum_{i \in V_{s}^{(p)}}\zeta_{p}^{\langle c, i\rangle_{s}}\chi_{u}(D_{F, i})\\
 &=\sum_{i \in V_{s}^{(p)} \backslash \{G(-u)\}}\zeta_{p}^{\langle c, i\rangle_{s}}(p^{n-s}\delta_{\{0\}}(u)-\varepsilon p^{\frac{n}{2}-s})+(p^{n-s}\delta_{\{0\}}(u)+\varepsilon p^{\frac{n}{2}-s}(p^{s}-1))\zeta_{p}^{\langle c, G(-u)\rangle_{s}}\\
 &=(p^{n-s}\delta_{\{0\}}(u)-\varepsilon p^{\frac{n}{2}-s})\sum_{i \in V_{s}^{(p)}}\zeta_{p}^{\langle c, i\rangle_{s}}+\varepsilon p^{\frac{n}{2}}\zeta_{p}^{G_{c}(-u)}\\
 &=\varepsilon p^{\frac{n}{2}}\zeta_{p}^{G_{c}(-u)}.
  \end{split}
\end{equation}
By (15) and the assumption that $\varepsilon=1$ if $p=2$, $F$ is a vectorial bent function with $\varepsilon_{F_{c}}=\varepsilon$ and $(F_{c})^{*}=G_{c}$ for any $c \in V_{s}^{(p)} \backslash \{0\}$.
Since $F_{c}$ is a weakly regular bent function, we have that $G_{c}=(F_{c})^{*}$ is also weakly regular bent and $G$ is vectorial bent. Thus, $F$ is vectorial dual-bent with $\varepsilon_{F_{c}}=\varepsilon$ and $(F_{c})^{*}=(F^{*})_{c}$ for any $c \in V_{s}^{(p)} \backslash \{0\}$, where $F^{*}=G$, that is, $F$ satisfies Condition A.
\end{proof}

Based on Lemma 1, we have the following theorem.

\begin{theorem}\label{1}
Let $F: V_{n}^{(p)}\rightarrow V_{s}^{(p)}$ be a vectorial dual-bent function satisfying Condition A. Define
\begin{equation*}
  A_{i}=D_{F, i}, \ i \in V_{s}^{(p)},
\end{equation*}
where $D_{F, i}=\{x \in V_{n}^{(p)}: F(x)=i\}$. Then $\Gamma=\{A_{i}, i \in V_{s}^{(p)}\}$ is a bent partition satisfying Condition $\mathcal{C}$.
\end{theorem}

\begin{proof}
By Lemma 1 and its proof, for any $i \in V_{s}^{(p)}$ and $u \in V_{n}^{(p)}$,
\begin{equation*}
  \chi_{u}(A_{i})=\chi_{u}(D_{F, i})=p^{n-s}\delta_{\{0\}}(u)+\varepsilon p^{\frac{n}{2}-s}(p^{s}\delta_{\{0\}}(F^{*}(-u)-i)-1),
\end{equation*}
where $\varepsilon=1$ if $p=2$ since all Boolean bent functions are regular. For any union $S$ of $p^{s-1}$ sets of $\{A_{i}: i \in V_{s}^{(p)}\}$, we have
\begin{equation}\label{16}
 \chi_{u}(S)=\left\{\begin{split}
                      p^{n-1}\delta_{\{0\}}(u)+\varepsilon p^{\frac{n}{2}-1}(p-1),  & \text{ if } A_{F^{*}(-u)}\subseteq S,\\
                      p^{n-1}\delta_{\{0\}}(u)-\varepsilon p^{\frac{n}{2}-1},  & \text{ if } A_{F^{*}(-u)}\nsubseteq S.
                    \end{split}
 \right.
\end{equation}
Let $f$ be an arbitrary function such that for every $j \in \mathbb{F}_{p}$, there are exactly $p^{s-1}$ sets $A_{i}$ in $\Gamma$ in its preimage. Define $g(u)=f(A_{F^{*}(-u)})$. Note that $g$ is a $p$-ary function from $V_{n}^{(p)}$ to $\mathbb{F}_{p}$. Then by (16), we have
\begin{equation}\label{17}
  \chi_{u}(D_{f, j})=\left\{\begin{split}
                               p^{n-1}\delta_{\{0\}}(u)+\varepsilon p^{\frac{n}{2}-1}(p-1), & \text{ if } j=g(u), \\
                               p^{n-1}\delta_{\{0\}}(u)-\varepsilon p^{\frac{n}{2}-1},  & \text{ if } j\neq g(u).
                            \end{split}\right.
\end{equation}
By (17), for any $u \in V_{n}^{(p)}$ we have
\begin{equation}\label{18}
 \begin{split}
    W_{f}(-u)& =\sum_{x \in V_{n}^{(p)}}\zeta_{p}^{f(x)+\langle u, x\rangle_{n}} \\
     & =\sum_{j \in \mathbb{F}_{p}}\zeta_{p}^{j}\sum_{x \in V_{n}^{(p)}: f(x)=j}\zeta_{p}^{\langle u, x\rangle_{n}}\\
     & =\sum_{j \in \mathbb{F}_{p}}\zeta_{p}^{j}\chi_{u}(D_{f, j})\\
   & =\sum_{j \in \mathbb{F}_{p}\backslash \{g(u)\}}\zeta_{p}^{j}(p^{n-1}\delta_{\{0\}}(u)-\varepsilon p^{\frac{n}{2}-1})+\zeta_{p}^{g(u)}(p^{n-1}\delta_{\{0\}}(u)+\varepsilon p^{\frac{n}{2}-1}(p-1))\\
   &=(p^{n-1}\delta_{\{0\}}(u)-\varepsilon p^{\frac{n}{2}-1})\sum_{j \in \mathbb{F}_{p}}\zeta_{p}^{j}+\varepsilon p^{\frac{n}{2}}\zeta_{p}^{g(u)}\\
   &=\varepsilon p^{\frac{n}{2}}\zeta_{p}^{g(u)}.
 \end{split}
\end{equation}
By (18) and $\varepsilon=1$ if $p=2$, $f$ is a weakly regular bent function with $\varepsilon_{f}=\varepsilon$ and $f^{*}(x)=g(-x)$. Let
\begin{equation*}
 W_{j}=\{u \in V_{n}^{(p)}: g(u)=j\}, j \in \mathbb{F}_{p},
\end{equation*}
then $W_{j}, j \in \mathbb{F}_{p}$ are pairwise disjoint and $\bigcup_{j \in \mathbb{F}_{p}}W_{j}=V_{n}^{(p)}$. By (17), for any $u \in V_{n}^{(p)}$ and nonempty set $J \subseteq \mathbb{F}_{p}$ we have
\begin{equation}\label{19}
  \chi_{u}(D_{f, J})=p^{n-1}\delta_{\{0\}}(u)|J|+\varepsilon p^{\frac{n}{2}-1}(p\delta_{W_{J}}(u)-|J|),
\end{equation}
where $D_{f, J}=\{x \in V_{n}^{(p)}: f(x) \in J\}$, $W_{J}=\bigcup_{j \in J}W_{j}$. By (19) and Lemma 1, $f$ is vectorial dual-bent with $(cf)^{*}=c(\beta f^{*}), c \in \mathbb{F}_{p}^{*}$ for some $\beta \in \mathbb{F}_{p}^{*}$ (since all vectorial duals of $f$ are $cf^{*}, c \in \mathbb{F}_{p}^{*}$). Let $c=1$, we obtain $\beta=1$, that is, $f$ is vectorial dual-bent with $(cf)^{*}=cf^{*}, c \in \mathbb{F}_{p}^{*}$. By Proposition 1, $f$ is a $(p-1)$-form. In particular, $F_{c}$ is a $(p-1)$-form for any $c \in \mathbb{F}_{p^s}^{*}$, which yields that $F(\alpha x)=F(x)$ for any $\alpha \in \mathbb{F}_{p}^{*}$ and $\mathbb{F}_{p}^{*}A_{i}=A_{i}, i \in V_{s}^{(p)}$. Hence, $\Gamma$ is a bent partition satisfying Condition $\mathcal{C}$.
\end{proof}

By Theorem 1, we have the following corollary.

\begin{corollary}\label{1}
Let $n$ be an even positive integer. Let $f: V_{n}^{(p)}\rightarrow \mathbb{F}_{p}$ be a weakly regular bent function of $(p-1)$-form, then $\{D_{f, j}, j \in \mathbb{F}_{p}\}$ is a bent partition of $V_{n}^{(p)}$, where $D_{f, j}=\{x \in V_{n}^{(p)}: f(x)=j\}$.
\end{corollary}

\begin{proof}
By Proposition 1, $f$ is a weakly regular vectorial dual-bent function with $(cf)^{*}=cf^{*}$. Since $n$ is even, $\varepsilon_{cf}=\varepsilon_{f}$ for all $c \in \mathbb{F}_{p}^{*}$ (see Theorem 1 of \cite{CM2013A}). Then by Theorem 1, the conclusion holds.
\end{proof}

A bent partition $\Gamma=\{A_{1}, \dots, A_{K}\}$ of depth $K$ is called \textit{coming from a normal bent partition} if there is $U\subseteq A_{i}$ for some $i$ such that $\{U, A_{1}, \dots, A_{i-1}, A_{i}\backslash U, A_{i+1}, \dots, A_{K}\}$ is a normal bent partition. In \cite{AM2022Be}, there is an open problem: Do bent partitions exist which are not coming from a normal bent partition of depth $K> 2$? In the following, we provide a positive answer for this open problem. By the definition of $l$-form, a ternary function $f$ is a $2$-form if and only if $f(x)=f(-x)$. Let $n$ be an even positive integer. If $f: V_{n}^{(3)}\rightarrow \mathbb{F}_{3}$ with $f(x)=f(-x)$ is a ternary weakly regular but not regular bent function (that is, $\varepsilon_{f}=-1$), then by Corollary 1, $\{D_{f, 0}, D_{f, 1}, D_{f, 2}\}$ is a bent partition of depth $3$. There exist such ternary bent functions $f$, for instance see \cite{CM2018Be,Meidl2022A}:
\begin{itemize}
  \item
  \begin{equation}\label{20}
    f(x)=Tr_{1}^{n}(\alpha x^{2}), x \in \mathbb{F}_{3^n},
 \end{equation}
  where $n$ is even, $\alpha \in \mathbb{F}_{3^n}^{*}$ is a square element if $4 \mid n$, and $\alpha \in \mathbb{F}_{3^n}^{*}$ is a non-square element if $4 \nmid n$;
  \item
  \begin{equation}\label{21}
    f(x)=Tr_{1}^{n}(a x^{\frac{3^n-1}{4}+3^m+1}), x \in \mathbb{F}_{3^n},
  \end{equation}
  where $n=2m$, $m$ odd, $a=\alpha^{\frac{3^m+1}{4}}$ for a primitive element $\alpha$ of $\mathbb{F}_{3^n}$;
  \item
  \begin{equation}\label{22}
    f(x)=Tr_{1}^{n}(\alpha(x^{3^{3k}+3^{2k}-3^{k}+1}+x^{2})), x \in \mathbb{F}_{3^n},
 \end{equation}
  where $n=4k$ for an arbitrary positive integer $k$, $\alpha \in \mathbb{F}_{3^{2k}}^{*}$;
  \item
  \begin{equation}\label{23}
    f(x, y, z)=(g(x)-h(x))z^{2}+yz+g(x), (x, y, z) \in \mathbb{F}_{3^n} \times \mathbb{F}_{3} \times \mathbb{F}_{3},
  \end{equation}
  where $n$ is even, $g$ and $h$ are distinct bent functions constructed by (20) or (22) if $4 \mid n$, $g$ and $h$ are distinct bent functions constructed by (20) or (21) if $4 \nmid n$.
\end{itemize}
For any ternary weakly regular but not regular bent function $f: V_{n}^{(3)}\rightarrow \mathbb{F}_{3}$ ($n$ even) with $f(x)=f(-x)$, the corresponding bent partition $\{D_{f, 0}, D_{f, 1}, D_{f, 2}\}$ is not coming from a normal bent partition by Theorem 4 (i) of \cite{AM2022Be}, which provides a positive answer for the above open problem proposed in \cite{AM2022Be}. We first recall Theorem 4 (i) of \cite{AM2022Be} and then give an example to illustrate this fact.

\begin{lemma}[\cite{AM2022Be}]\label{2}
Let $\Gamma=\{U, A_{1}, \dots, A_{K}\}$ be a normal bent partition of $V_{n}^{(p)}$. Then $|U|=p^{\frac{n}{2}}$ and $|A_{j}|=\frac{p^{n}-p^{\frac{n}{2}}}{K}, 1\leq j \leq K$.
\end{lemma}

\begin{example}\label{1}
Let $f: \mathbb{F}_{3^4}\rightarrow \mathbb{F}_{3}$ be defined by $f(x)=Tr_{1}^{4}(x^{2})$. Then $f$ is ternary weakly regular bent with $f(x)=f(-x)$ and $\varepsilon_{f}=-1$. By Corollary 1, $\{D_{f, 0}, D_{f, 1}, D_{f, 2}\}$ is a bent partition. By the result of Nyberg \cite{Nyberg1991Co}, for any weakly regular $p$-ary bent function $g: V_{n}^{(p)}\rightarrow \mathbb{F}_{p}$ with $n$ even, we have $\{|D_{g, i}|, i \in \mathbb{F}_{p}\}=\{p^{n-1}+\varepsilon_{f}p^{\frac{n}{2}-1}(p-1), p^{n-1}-\varepsilon_{f}p^{\frac{n}{2}-1}\}$. For our example, $|D_{f, 0}|=21$, $|D_{f, 1}|=|D_{f, 2}|=30$. By Lemma 2, it is easy to see that $\{D_{f, 0}, D_{f, 1}, D_{f, 2}\}$ can not be from a normal bent partition.
\end{example}

In the following, based on Theorem 1, we give an alternative proof that $\Gamma_{i}, \Gamma_{i}^{\bullet}, \Theta_{i}, i=1,2$ defined by (5)-(10) given in \cite{AKM2022Ge} are bent partitions.

Let $s, n$ be positive integers with $s \mid n$, $u$ be an integer with $u\equiv p^{j_{0}} \ mod \ (p^{s}-1)$ for some $0\leq j_{0}\leq s-1$ and $gcd(u, p^{n}-1)=1$, and let $d$ be an integer with $du\equiv 1 \ mod \ (p^{n}-1)$. Let $P=(\mathbb{F}_{p^n}, +, \circ)$ be a (pre)semifield such that its dual $P^{\star}=(\mathbb{F}_{p^n}, +, \star)$ is right $\mathbb{F}_{p^s}$-linear. For given $x \in \mathbb{F}_{p^n}$, if $x=0$, then let $\eta_{x}=0$, and if $x\neq 0$, then let $\eta_{x}$ be given by $x\star \eta_{x}^{-1}=1$ (For convention we set $\eta_{0}^{-1}=\eta_{0}^{p^n-2}=0$). For any $\alpha \in \mathbb{F}_{p^n}^{*}$ and $i_{0} \in \mathbb{F}_{p^s}$, define
\begin{equation}\label{24}
  F(x, y)=Tr_{s}^{n}(\alpha a(x, y))+i_{0}(1-x^{p^n-1}), (x, y)\in
\mathbb{F}_{p^n} \times \mathbb{F}_{p^n},
\end{equation}
where for given $(x, y)$, if $x=0$, then $a(x, y)=0$, and if $x \neq 0$, then $a(x, y)$ is given by $a(x, y)\circ x^{u}=y$,
and
\begin{equation}\label{25}
  F^{\bullet}(x, y)=Tr_{s}^{n}(\alpha a^{\bullet}(x, y))+i_{0}(1-x^{p^n-1}), (x, y)\in
\mathbb{F}_{p^n} \times \mathbb{F}_{p^n},
\end{equation}
where for given $(x, y)$, if $x=0$, then $a^{\bullet}(x, y)=0$, and if $x \neq 0$, then $a^{\bullet}(x, y)$ is given by $x^{u}\circ a^{\bullet}(x, y)=y$,
and
\begin{equation}\label{26}
  G(x, y)=Tr_{s}^{n}(\alpha b(x, y))+i_{0}(1-y^{p^n-1}), (x, y)\in
\mathbb{F}_{p^n} \times \mathbb{F}_{p^n},
\end{equation}
where for given $(x, y)$, if $y=0$, then $b(x, y)=0$, and if $y\neq 0$, then $b(x, y)$ is given by $b(x, y)\circ y^{d}=x$,
and
\begin{equation}\label{27}
  G^{\bullet}(x, y)=Tr_{s}^{n}(\alpha b^{\bullet}(x, y))+i_{0}(1-y^{p^n-1}), (x, y)\in
\mathbb{F}_{p^n} \times \mathbb{F}_{p^n},
\end{equation}
where for given $(x, y)$, if $y=0$, then $b^{\bullet}(x, y)=0$, and if $y\neq 0$, then $b^{\bullet}(x, y)$ is given by $y^{d}\circ b^{\bullet}(x, y)=x$,
and
\begin{equation}\label{28}
  M(x, y)=Tr_{s}^{n}(\alpha \eta_{x}^{-u}y)+i_{0}(1-x^{p^n-1}), (x, y)\in
\mathbb{F}_{p^n} \times \mathbb{F}_{p^n},
\end{equation}
and
\begin{equation}\label{29}
  N(x, y)=Tr_{s}^{n}(\alpha x\eta_{y}^{-d})+i_{0}(1-y^{p^n-1}), (x, y)\in
\mathbb{F}_{p^n} \times \mathbb{F}_{p^n}.
\end{equation}

\begin{proposition}\label{3}
Let $F, F^{\bullet}, G, G^{\bullet}, M, N$ be defined as above. Then they are all vectorial dual-bent functions satisfying Condition A with $\varepsilon=1$.
\end{proposition}

\begin{proof}
We only prove the result for $F$ and $M$ since the proofs for $F^{\bullet}, G, G^{\bullet}$ are similar to the proof for $F$, and the proof for $N$ is similar to the proof for $M$.
\begin{itemize}
  \item For $F$:
\end{itemize}
For any $c \in \mathbb{F}_{p^s}^{*}$, we have
\begin{equation*}
  F_{c}(x, y)=Tr_{1}^{n}(c\alpha a(x, y))+Tr_{1}^{s}(ci_{0})(1-x^{p^{n}-1}).
\end{equation*}
For any $c \in \mathbb{F}_{p^s}^{*}$ and $(w, v) \in \mathbb{F}_{p^n} \times \mathbb{F}_{p^n}$, we have
\begin{equation*}
 \begin{split}
    W_{F_{c^u}}(w, v) & =\sum_{x \in \mathbb{F}_{p^n}^{*}} \sum_{y \in \mathbb{F}_{p^n}}\zeta_{p}^{Tr_{1}^{n}(c^u \alpha a(x, y))-Tr_{1}^{n}(wx+vy)}+\zeta_{p}^{Tr_{1}^{s}(c^ui_{0})}\sum_{y \in \mathbb{F}_{p^n}}\zeta_{p}^{-Tr_{1}^{n}(vy)}\\
    & =\sum_{x \in \mathbb{F}_{p^n}} \sum_{y \in \mathbb{F}_{p^n}}\zeta_{p}^{Tr_{1}^{n}(c^u \alpha a(x, y))-Tr_{1}^{n}(wx+vy)}+p^{n}(\zeta_{p}^{Tr_{1}^{s}(c^ui_{0})}-1)\delta_{\{0\}}(v)\\
    & =W_{h}(w, v)+p^{n}(\zeta_{p}^{Tr_{1}^{s}(c^ui_{0})}-1)\delta_{\{0\}}(v),\\
 \end{split}
\end{equation*}
where $h(x, y)=Tr_{1}^{n}(c^u \alpha a(x, y))$. For given $x \in \mathbb{F}_{p^n}$, if $x=0$, then let $\lambda_{x}=0$, and if $x\neq 0$, then let $\lambda_{x}$ be given by $x \star \lambda_{x}^{-1}=\alpha$ (For convention we set $\lambda_{0}^{-1}=\lambda_{0}^{p^n-2}=0$). Define $\rho(x)=\lambda_{x}^{-d}$. Then $\rho$ is a permutation over $\mathbb{F}_{p^n}$. For any $x \in \mathbb{F}_{p^n}^{*}$, set $z=\rho^{-1}(c^{-1}x)$. Then $\lambda_{z}^{-d}=\rho(z)=c^{-1}x$. By $du\equiv 1 \ mod \ (p^n-1)$, we have $\lambda_{z}^{-1}=c^{-u}x^{u}$. Since $z \neq 0$ and $P^{\star}$ is right $\mathbb{F}_{p^s}$-linear, we have $\alpha=z\star \lambda_{z}^{-1}=z\star (c^{-u}x^{u})=c^{-u}(z\star x^{u})$, that is, $\rho^{-1}(c^{-1}x)\star x^{u}=\alpha c^{u}$ for any $x \neq 0$. Thus, when $x \neq 0$, $Tr_{1}^{n}(c^u\alpha a(x, y))=Tr_{1}^{n}(a(x, y)(\rho^{-1}(c^{-1}x)\star x^{u}))=Tr_{1}^{n}(\rho^{-1}(c^{-1}x)(a(x, y)\circ x^{u}))=Tr_{1}^{n}(\rho^{-1}(c^{-1}x)y)$. When $x=0$, by $a(0, y)=\rho^{-1}(0)=0$, we have $Tr_{1}^{n}(c^u\alpha a(x, y))=Tr_{1}^{n}(\rho^{-1}(c^{-1}x)y)=0$. Hence, $h(x, y)=Tr_{1}^{n}(\rho^{-1}(c^{-1}x)y)$, which is a Maiorana-McFarland bent function and by (4),
\begin{equation*}
  W_{h}(w, v)=p^{n}\zeta_{p}^{-Tr_{1}^{n}(c w \rho(v))}.
\end{equation*}
Therefore, for any $c \in \mathbb{F}_{p^s}^{*}$,
\begin{equation}\label{30}
  \begin{split}
  W_{F_{c^u}}(w, v)&=p^{n}(\zeta_{p}^{-Tr_{1}^{n}(c w \rho(v))}+(\zeta_{p}^{Tr_{1}^{s}(c^ui_{0})}-1)\delta_{\{0\}}(v))\\
  &=p^{n}\zeta_{p}^{-Tr_{1}^{n}(cw \rho(v))+Tr_{1}^{s}(c^ui_{0})(1-v^{p^n-1})}.\\
  \end{split}
\end{equation}
By (30) and $ud\equiv 1 \ mod \ (p^n-1)$, we have that for any $c \in \mathbb{F}_{p^s}^{*}$, $F_{c}$ is a regular bent function with
\begin{equation*}
  \begin{split}
  (F_{c})^{*}(x, y)&=-Tr_{1}^{n}(c^dx\rho(y))+Tr_{1}^{s}(ci_{0})(1-y^{p^n-1})\\
  & =-Tr_{1}^{n}(c^{dp^{j_{0}}}(x\rho(y))^{p^{j_{0}}})+Tr_{1}^{s}(ci_{0})(1-y^{p^n-1}).
  \end{split}
\end{equation*}
Since $u\equiv p^{j_{0}} \ mod \ (p^{s}-1)$ and $du\equiv 1 \ mod \ (p^n-1)$, we have $d\equiv p^{s-j_{0}} \ mod \ (p^{s}-1)$ and thus $(c^{d})^{p^{j_{0}}}=c$ for any $c \in \mathbb{F}_{p^s}^{*}$. Therefore, $F$ is a vectorial bent function with $\varepsilon_{F_{c}}=1$ and $(F_{c})^{*}=H_{c}$ for all $c \in \mathbb{F}_{p^s}^{*}$, where
\begin{equation*}
     H(x, y)=-Tr_{s}^{n}((x\rho(y))^{p^{j_{0}}})+i_{0}(1-y^{p^n-1})=-(Tr_{s}^{n}(x\rho(y)))^{p^{j_{0}}}+i_{0}(1-y^{p^n-1}).
\end{equation*}
Since $F_{c}$ is regular bent, we have that $(F_{c})^{*}=H_{c}$ is also regular bent and $H$ is vectorial bent. Thus, $F$ is vectorial dual-bent with $\varepsilon_{F_{c}}=1$ and $(F_{c})^{*}=(F^{*})_{c}$ for all $c \in \mathbb{F}_{p^s}^{*}$, where $F^{*}=H$, that is, $F$ satisfies Condition A.
\begin{itemize}
  \item For $M$:
\end{itemize}
For any $c \in \mathbb{F}_{p^s}^{*}$,
\begin{equation*}
  M_{c}(x, y)=Tr_{1}^{n}(c \alpha \eta_{x}^{-u}y)+Tr_{1}^{s}(ci_{0})(1-x^{p^n-1}).
\end{equation*}
Similar to the discussion for $F$, for any $c \in \mathbb{F}_{p^s}^{*}$ and $(w, v) \in \mathbb{F}_{p^n} \times \mathbb{F}_{p^n}$ we have
\begin{equation*}
  W_{M_{c}}(w, v)=W_{g}(w, v)+p^{n}(\zeta_{p}^{Tr_{1}^{s}(ci_{0})}-1)\delta_{\{0\}}(v),
\end{equation*}
where $g(x, y)=Tr_{1}^{n}(c\alpha\eta_{x}^{-u}y)$. Let $\pi(x)=\eta_{x}^{-u}$, then $\pi$ is a permutation over $\mathbb{F}_{p^n}$. Since $g$ is a Maiorana-McFarland bent function, then by (4),
\begin{equation*}
  W_{g}(w, v)=p^{n}\zeta_{p}^{-Tr_{1}^{n}(w\pi^{-1}(c^{-1}\alpha^{-1}v))}.
\end{equation*}
For any given $y \in \mathbb{F}_{p^n}^{*}$, set $\pi^{-1}(c^{-1}\alpha^{-1}y)=z$. Then $c^{-1}\alpha^{-1}y=\pi(z)=\eta_{z}^{-u}$. By $ud\equiv 1 \ mod \ (p^n-1)$, we have $\eta_{z}^{-1}=c^{-d}\alpha^{-d}y^{d}$. Since $z\neq 0$ and $P^{\star}$ is right $\mathbb{F}_{p^s}$-linear, we have $1=z\star \eta_{z}^{-1}=z\star (c^{-d}\alpha^{-d}y^{d})=c^{-d}(z\star \alpha^{-d}y^{d})$, that is, $\pi^{-1}(c^{-1}\alpha^{-1}y)\star \alpha^{-d}y^{d}=c^{d}$. For given $(x, y) \in \mathbb{F}_{p^n} \times \mathbb{F}_{p^n}$, if $y=0$, then let $r(x, y)=0$, and if $y\neq 0$, then let $r(x, y)$ be given by $r(x, y)\circ \alpha^{-d}y^{d}=x$. When $v\neq 0$, we have $Tr_{1}^{n}(w\pi^{-1}(c^{-1}\alpha^{-1}v))=Tr_{1}^{n}(\pi^{-1}(c^{-1}\alpha^{-1}v)(r(w, v)\circ \alpha^{-d}v^{d}))=Tr_{1}^{n}(r(w, v)(\pi^{-1}(c^{-1}\alpha^{-1}v)\star \alpha^{-d}v^{d}))=Tr_{1}^{n}(c^{d}r(w, v))=Tr_{1}^{n}(c(r(w, v))^{p^{j_{0}}})$. When $v=0$, since $\pi^{-1}(0)=0$ and $r(w, 0)=0$, we have $Tr_{1}^{n}(w\pi^{-1}(c^{-1}\alpha^{-1}v))=Tr_{1}^{n}(c(r(w, v))^{p^{j_{0}}})=0$. Thus, $-Tr_{1}^{n}(w\pi^{-1}(c^{-1}\alpha^{-1}v))=-Tr_{1}^{n}(c(r(w, v))^{p^{j_{0}}})$ and
\begin{equation*}
 \begin{split}
  W_{M_{c}}(w, v)&=p^{n}(\zeta_{p}^{-Tr_{1}^{n}(c(r(w, v))^{p^{j_{0}}})}+(\zeta_{p}^{Tr_{1}^{s}(ci_{0})}-1)\delta_{\{0\}}(v))\\
  &=p^{n}\zeta_{p}^{-Tr_{1}^{n}(c(r(w, v))^{p^{j_{0}}})+Tr_{1}^{s}(ci_{0})(1-v^{p^n-1})},\\
  \end{split}
\end{equation*}
which implies that $M$ is a vectorial dual-bent function with $\varepsilon_{M_{c}}=1$ and $(M_{c})^{*}=(M^{*})_{c}$ for all $c \in \mathbb{F}_{p^s}^{*}$, where
\begin{equation*}
  M^{*}(x, y)=-Tr_{s}^{n}((r(x, y))^{p^{j_{0}}})+i_{0}(1-y^{p^n-1}),
\end{equation*}
that is, $M$ satisfies Condition A.
\end{proof}

By Theorem 1 and Proposition 3, we have that $\{D_{F, i}, i \in \mathbb{F}_{p^s}\}$, $\{D_{F^{\bullet}, i}, i \in \mathbb{F}_{p^s}\}$, $\{D_{G, i}, i \in \mathbb{F}_{p^s}\}$, $\{D_{G^{\bullet}, i}, i \in \mathbb{F}_{p^s}\}$, $\{D_{M, i}, i \in \mathbb{F}_{p^s}\}$ and $\{D_{N, i}, i \in \mathbb{F}_{p^s}\}$ are bent partitions. It is easy to verify that
\begin{small}
\begin{equation*}
\begin{split}
 & D_{F, i}=\left\{\begin{split}
                  \bigcup_{t \in \mathbb{F}_{p^n}: Tr_{s}^{n}(\alpha t)=i}U_{t}, & \ \text{ if } i\neq i_{0}, \\
                  \bigcup_{t \in \mathbb{F}_{p^n}: Tr_{s}^{n}(\alpha t)=i_{0}}U_{t}\bigcup U, & \ \text{ if } i=i_{0},\\
               \end{split}\right.,
 \ D_{F^{\bullet}, i}=\left\{\begin{split}
                  \bigcup_{t \in \mathbb{F}_{p^n}: Tr_{s}^{n}(\alpha t)=i}U_{t}^{\bullet}, & \ \text{ if } i\neq i_{0}, \\
                  \bigcup_{t \in \mathbb{F}_{p^n}: Tr_{s}^{n}(\alpha t)=i_{0}}U_{t}^{\bullet}\bigcup U, & \ \text{ if } i=i_{0},\\
               \end{split}\right.,\\
  &D_{G, i}=\left\{\begin{split}
                  \bigcup_{t \in \mathbb{F}_{p^n}: Tr_{s}^{n}(\alpha t)=i}V_{t}, & \ \text{ if } i\neq i_{0}, \\
                  \bigcup_{t \in \mathbb{F}_{p^n}: Tr_{s}^{n}(\alpha t)=i_{0}}V_{t}\bigcup V, & \ \text{ if } i=i_{0},\\
               \end{split}\right.,
  \ D_{G^{\bullet}, i}=\left\{\begin{split}
                  \bigcup_{t \in \mathbb{F}_{p^n}: Tr_{s}^{n}(\alpha t)=i}V_{t}^{\bullet}, & \ \text{ if } i\neq i_{0}, \\
                  \bigcup_{t \in \mathbb{F}_{p^n}: Tr_{s}^{n}(\alpha t)=i_{0}}V_{t}^{\bullet}\bigcup V, & \ \text{ if } i=i_{0},\\
               \end{split}\right.,\\
  &D_{M, i}=\left\{\begin{split}
                  \bigcup_{t \in \mathbb{F}_{p^n}: Tr_{s}^{n}(\alpha t)=i}X_{t}, & \ \text{ if } i\neq i_{0}, \\
                  \bigcup_{t \in \mathbb{F}_{p^n}: Tr_{s}^{n}(\alpha t)=i_{0}}X_{t}\bigcup X, & \ \text{ if } i=i_{0},\\
               \end{split}\right.,
  \ D_{N, i}=\left\{\begin{split}
                  \bigcup_{t \in \mathbb{F}_{p^n}: Tr_{s}^{n}(\alpha t)=i}Y_{t}, & \ \text{ if } i\neq i_{0}, \\
                  \bigcup_{t \in \mathbb{F}_{p^n}: Tr_{s}^{n}(\alpha t)=i_{0}}Y_{t}\bigcup Y, & \ \text{ if } i=i_{0},\\
               \end{split}\right.
\end{split}
\end{equation*}
\end{small}where
\begin{equation*}
\begin{split}
& U_{t}=\{(x, t\circ x^{u}): x \in \mathbb{F}_{p^n}^{*}\} \text{ if } t \in \mathbb{F}_{p^n}, \text{ and }U= \{(0, y): y \in \mathbb{F}_{p^n}\},\\
&U_{t}^{\bullet}=\{(x, x^{u}\circ t): x \in \mathbb{F}_{p^n}^{*}\} \text{ if } t \in \mathbb{F}_{p^n}, \text{ and }U= \{(0, y): y \in \mathbb{F}_{p^n}\},\\
&V_{t}=\{(t\circ x^{d}, x): x \in \mathbb{F}_{p^n}^{*}\} \text{ if } t \in \mathbb{F}_{p^n}, \text{ and }V=\{(x, 0):  x \in \mathbb{F}_{p^n}\},\\
&V_{t}^{\bullet}=\{(x^{d}\circ t, x): x \in \mathbb{F}_{p^n}^{*}\} \text{ if } t \in \mathbb{F}_{p^n}, \text{ and }V=\{(x, 0):  x \in \mathbb{F}_{p^n}\},\\
& X_{t}=\{(t\eta_{x}^{d}, x): x \in \mathbb{F}_{p^n}^{*}\} \text{ if } t \in \mathbb{F}_{p^n}, \text{ and }X=\{(x, 0): x \in \mathbb{F}_{p^n}\},\\
& Y_{t}=\{(x, t\eta_{x}^{u}): x \in \mathbb{F}_{p^n}^{*}\} \text{ if } t \in \mathbb{F}_{p^n}, \text{ and }Y=\{(0, y): y \in \mathbb{F}_{p^n}\}.
\end{split}
\end{equation*}
For the above bent partitions from vectorial dual-bent functions $F, F^{\bullet}, G, G^{\bullet}, M, N$, by setting $\alpha=1, u=p^s+p-1$ with $gcd(u, p^n-1)=1$, then we can obtain bent partitions $\Gamma_{1}, \Gamma_{1}^{\bullet}, \Gamma_{2}, \Gamma_{2}^{\bullet}, \Theta_{1}, \Theta_{2}$ defined by (5)-(10) respectively. Thus by the above analysis, we provide an alternative derivation that $\Gamma_{1}, \Gamma_{1}^{\bullet}, \Gamma_{2}, \Gamma_{2}^{\bullet}, \Theta_{1}, \Theta_{2}$ are bent partitions.

When $p$ is an odd prime, we show that the converse of Theorem 1 also holds.

\begin{theorem}\label{2}
Let $p$ be an odd prime. Let $\Gamma=\{A_{i}, i \in V_{s}^{(p)}\}$ be a bent partition of $V_{n}^{(p)}$ satisfying Condition $\mathcal{C}$. Define $F: V_{n}^{(p)}\rightarrow V_{s}^{(p)}$ by
\begin{equation*}
  F(x)=i \text{ if } x \in A_{i}.
\end{equation*}
Then $F$ is a vectorial dual-bent function satisfying Condition A.
\end{theorem}

\begin{proof}
Since $\mathbb{F}_{p}^{*}A_{i}=A_{i}$ for any $i \in V_{s}^{(p)}$, all bent functions constructed from $\Gamma$ are $(p-1)$-form. When $s=1$, the conclusion follows from Proposition 1. In the following, we consider the case of $s\geq 2$.

Let $f$ be an arbitrary bent function constructed from $\Gamma$. By Lemma 3.4 of \cite{HLL2020Ra}, for any $u \in V_{n}^{(p)}$ and $j \in \mathbb{F}_{p}$ we have
\begin{equation}\label{31}
  \chi_{u}(D_{f, j})=\left\{\begin{split}
                               p^{n-1}\delta_{\{0\}}(u)+\varepsilon p^{\frac{n}{2}-1}(p-1), & \text{ if } f^{*}(u)=j, \\
                               p^{n-1}\delta_{\{0\}}(u)-\varepsilon p^{\frac{n}{2}-1}, &  \text{ if } f^{*}(u) \neq j,
                            \end{split}\right.
\end{equation}
where $D_{f, j}=\{x \in V_{n}^{(p)}: f(x)=j\}, j \in \mathbb{F}_{p}$. For any fixed $u \in V_{n}^{(p)}$, since
 \begin{equation*}
   \{\chi_{u}(D_{f, j}), j \in \mathbb{F}_{p}\}=\{p^{n-1}\delta_{\{0\}}(u)+\varepsilon p^{\frac{n}{2}-1}(p-1), p^{n-1}\delta_{\{0\}}(u)-\varepsilon p^{\frac{n}{2}-1}\}
\end{equation*}
for any bent function $f$ constructed from $\Gamma$, we have that for any fixed $u \in V_{n}^{(p)}$, there exists a unique $G(u) \in V_{s}^{(p)}$ such that $\chi_{u}(A_{i}), i\neq G(u)$ are all the same and $\chi_{u}(A_{i})\neq \chi_{u}(A_{G(u)}), i\neq G(u)$. Note that $G$ is a function from $V_{n}^{(p)}$ to $V_{s}^{(p)}$. Moreover by (31), for any fixed $u \in V_{n}^{(p)}$ we have
\begin{equation}\label{32}
  \chi_{u}(A_{i})=\left\{\begin{split}
                            p^{n-s}\delta_{\{0\}}(u)+\varepsilon p^{\frac{n}{2}-s}(p^{s}-1), & \text{ if } i=G(u), \\
                            p^{n-s}\delta_{\{0\}}(u)- \varepsilon p^{\frac{n}{2}-s}, & \text{ if } i\neq G(u).
                         \end{split}\right.
\end{equation}
Define
\begin{equation*}
  W_{i}=\{u \in V_{n}^{(p)}: G(u)=i\}, i \in V_{s}^{(p)}.
\end{equation*}
Then obviously $W_{i}, i \in V_{s}^{(p)}$ are pairwise disjoint and $\bigcup_{i \in V_{s}^{(p)}}W_{i}=V_{n}^{(p)}$. By (32), for any $u \in V_{n}^{(p)}$ and nonempty set $I \subseteq V_{s}^{(p)}$ we have
\begin{equation}\label{33}
  \chi_{u}(D_{F, I})=\sum_{i \in I}\chi_{u}(A_{i})=p^{n-s}\delta_{\{0\}}(u)|I|+\varepsilon p^{\frac{n}{2}-s}(p^{s}\delta_{W_{I}}(u)-|I|),
\end{equation}
where $D_{F, I}=\{x \in V_{n}^{(p)}: F(x) \in I\}$, $W_{I}=\bigcup_{i \in I}W_{i}$. By (33) and Lemma 1, the conclusion holds.
\end{proof}

When $p$ is an odd prime, from Theorems 1 and 2 we obtain a characterization of bent partitions satisfying Condition $\mathcal{C}$ in terms of vectorial dual-bent functions.

\begin{theorem}\label{3}
Let $p$ be an odd prime. Let $\Gamma=\{A_{i}, i \in V_{s}^{(p)}\}$ be a partition of $V_{n}^{(p)}$, where $n$ is even, $s\leq \frac{n}{2}$. Define $F: V_{n}^{(p)}\rightarrow V_{s}^{(p)}$ as
\begin{equation*}
  F(x)=i \text{ if } x \in A_{i}.
\end{equation*}
Then $\Gamma$ is a bent partition satisfying Condition $\mathcal{C}$ if and only if $F$ is a vectorial dual-bent function satisfying Condition A.
\end{theorem}

\section{Constructing bent partitions from vectorial dual-bent functions}
\label{sec:4}
In this section, we construct bent partitions from vectorial dual-bent functions.

The following theorem provides a secondary construction of vectorial dual-bent functions, which can be used to generate more bent partitions.

\begin{theorem}\label{4}
Let $n, m, s$ be positive integers for which $n$ is even and $s\leq \frac{n}{2}, s \mid m, s\neq m$. For any $i \in \mathbb{F}_{p^s}$, let $F(i; x): V_{n}^{(p)}\rightarrow \mathbb{F}_{p^s}$ be a vectorial dual-bent function with $((F(i; x))_{c})^{*}=((F(i; x))^{*})_{c}$ and $\varepsilon_{(F(i; x))_{c}}=\varepsilon$ for any $c \in \mathbb{F}_{p^s}^{*}$, where $(F(i; x))^{*}$ is a vectorial dual of $F(i; x)$ and $\varepsilon \in \{\pm 1\}$ is a constant independent of $i, c$. Let $\alpha, \beta \in \mathbb{F}_{p^m}$ be linearly independent over $\mathbb{F}_{p^s}$. Let $R$ be a permutation over $\mathbb{F}_{p^m}$ with $R(0)=0$ and $T: \mathbb{F}_{p^s}\rightarrow \mathbb{F}_{p^s}$ be an arbitrary function. Define $H: V_{n}^{(p)} \times \mathbb{F}_{p^m} \times \mathbb{F}_{p^m}\rightarrow \mathbb{F}_{p^s}$ as
\begin{small}
\begin{equation*}
  H(x, y_{1}, y_{2})=F(Tr_{s}^{m}(\alpha R(y_{1}y_{2}^{p^m-2})); x)+Tr_{s}^{m}(\beta R(y_{1}y_{2}^{p^m-2}))+T(Tr_{s}^{m}(\alpha R(y_{1}y_{2}^{p^m-2}))).
\end{equation*}
\end{small}Then $H$ is a vectorial dual-bent function satisfying Condition A and $\Gamma=\{A_{i}, i \in \mathbb{F}_{p^s}\}$ is a bent partition satisfying Condition $\mathcal{C}$, where $A_{i}=\{(x, y_{1}, y_{2}) \in V_{n}^{(p)} \times \mathbb{F}_{p^m} \times \mathbb{F}_{p^m}: H(x, y_{1}, y_{2})=i\}$.
\end{theorem}

\begin{proof}
Denote
\begin{equation*}
  d(y)=Tr_{s}^{m}(\beta R(y_{1}y_{2}^{p^m-2})), e(y)=Tr_{s}^{m}((\beta-\alpha) R(y_{1}y_{2}^{p^m-2})), y=(y_{1}, y_{2}) \in \mathbb{F}_{p^m} \times \mathbb{F}_{p^m}.
\end{equation*}
For any $c \in \mathbb{F}_{p^s}^{*}$ and $(a, b)=(a, b_{1}, b_{2}) \in V_{n}^{(p)} \times \mathbb{F}_{p^m} \times \mathbb{F}_{p^m}$, we have
\begin{small}
\begin{equation*}
\begin{split}
  & W_{H_{c}}(a, b)\\
   & = \sum_{x \in V_{n}^{(p)}}\sum_{y=(y_{1}, y_{2}) \in \mathbb{F}_{p^m} \times \mathbb{F}_{p^m}}\zeta_{p}^{Tr_{1}^{s}(cF(d(y)-e(y); x))+Tr_{1}^{s}(cd(y))+Tr_{1}^{s}(cT(d(y)-e(y)))}\zeta_{p}^{-\langle a, x\rangle_{n}-Tr_{1}^{m}(b_{1}y_{1}+b_{2}y_{2})}\\
 & =\sum_{i \in \mathbb{F}_{p^s}}\sum_{y=(y_{1}, y_{2}) \in \mathbb{F}_{p^m} \times \mathbb{F}_{p^m}: d(y)-e(y)=i}\sum_{x \in V_{n}^{(p)}}\zeta_{p}^{Tr_{1}^{s}(cF(i; x))+Tr_{1}^{s}(cd(y))+Tr_{1}^{s}(cT(i))}\zeta_{p}^{-\langle a, x\rangle_{n}-Tr_{1}^{m}(b_{1}y_{1}+b_{2}y_{2})}\\
 & =p^{-s}\sum_{i \in \mathbb{F}_{p^s}}W_{(F(i; x))_{c}}(a)\zeta_{p}^{Tr_{1}^{s}(cT(i))}\sum_{y=(y_{1}, y_{2}) \in \mathbb{F}_{p^m} \times \mathbb{F}_{p^m}}\zeta_{p}^{Tr_{1}^{s}(cd(y))-Tr_{1}^{m}(b_{1}y_{1}+b_{2}y_{2})}\sum_{j \in \mathbb{F}_{p^s}}\zeta_{p}^{Tr_{1}^{s}(cj(i-(d(y)-e(y))))}\\
 & =p^{-s}\sum_{i \in \mathbb{F}_{p^s}}W_{(F(i; x))_{c}}(a)\zeta_{p}^{Tr_{1}^{s}(cT(i))}\sum_{j \in \mathbb{F}_{p^s}}\zeta_{p}^{Tr_{1}^{s}(ijc)}\sum_{y=(y_{1}, y_{2}) \in \mathbb{F}_{p^m} \times \mathbb{F}_{p^m}}\zeta_{p}^{Tr_{1}^{s}(c((1-j)d(y)+je(y)))-Tr_{1}^{m}(b_{1}y_{1}+b_{2}y_{2})}.\\
\end{split}
\end{equation*}
\end{small}By Theorem 3 of \cite{CMP2018Ve}, for any $j \in \mathbb{F}_{p^s}$, $J(j; y)=(1-j)d(y)+je(y)$ is a partial spread vectorial dual-bent function with $\varepsilon_{(J(j; y))_{c}}=1$ and $((J(j; y))_{c})^{*}=((1-j)d^{*}(y)+je^{*}(y))_{c}$ for any $c \in \mathbb{F}_{p^s}^{*}$, where $d^{*}(y)=Tr_{s}^{m}(\beta R(-y_{1}^{p^m-2}y_{2}))$, $e^{*}(y)=Tr_{s}^{m}((\beta-\alpha)R(-y_{1}^{p^m-2}y_{2}))$. Therefore,
\begin{small}
\begin{equation}\label{34}
   \begin{split}
     & W_{H_{c}}(a, b) \\
     & =p^{m-s}\sum_{i \in \mathbb{F}_{p^s}}W_{(F(i; x))_{c}}(a)\zeta_{p}^{Tr_{1}^{s}(cT(i))}\sum_{j \in \mathbb{F}_{p^s}}\zeta_{p}^{Tr_{1}^{s}(ijc)}\zeta_{p}^{Tr_{1}^{s}(c((1-j)d^{*}(b)+je^{*}(b)))}\\
    & =p^{m-s}\zeta_{p}^{Tr_{1}^{s}(cd^{*}(b))}\sum_{i \in \mathbb{F}_{p^s}}W_{(F(i; x))_{c}}(a)\zeta_{p}^{Tr_{1}^{s}(cT(i))}\sum_{j \in \mathbb{F}_{p^s}}\zeta_{p}^{Tr_{1}^{s}(cj(i-(d^{*}(b)-e^{*}(b))))}\\
    & =p^{m}\zeta_{p}^{Tr_{1}^{s}(cd^{*}(b))}W_{(F(d^{*}(b)-e^{*}(b); x))_{c}}(a)\zeta_{p}^{Tr_{1}^{s}(cT(d^{*}(b)-e^{*}(b)))}\\
    & =\varepsilon p^{\frac{n}{2}+m}\zeta_{p}^{((F(Tr_{s}^{m}(\alpha R(-b_{1}^{p^m-2}b_{2})); x))_{c})^{*}(a)+Tr_{1}^{s}(cTr_{s}^{m}(\beta R(-b_{1}^{p^m-2}b_{2})))+Tr_{1}^{s}(c T(Tr_{s}^{m}(\alpha R(-b_{1}^{p^m-2}b_{2}))))}\\
    & =\varepsilon p^{\frac{n}{2}+m}\zeta_{p}^{((F(Tr_{s}^{m}(\alpha R(-b_{1}^{p^m-2}b_{2})); x))^{*})_{c}(a)+Tr_{1}^{s}(cTr_{s}^{m}(\beta R(-b_{1}^{p^m-2}b_{2})))+Tr_{1}^{s}(c T(Tr_{s}^{m}(\alpha R(-b_{1}^{p^m-2}b_{2}))))}.\\
   \end{split}
\end{equation}
\end{small}Note that $\varepsilon=1$ if $p=2$ since all Boolean bent functions are regular. By (34), $H$ is a vectorial bent function with $(H_{c})^{*}=G_{c}$ and $\varepsilon_{H_{c}}=\varepsilon$ for any $c \in \mathbb{F}_{p^s}^{*}$, where
\begin{small}
\begin{equation*}
  G(a, b_{1}, b_{2})=(F(Tr_{s}^{m}(\alpha R(-b_{1}^{p^m-2}b_{2})); x))^{*}(a)+Tr_{s}^{m}(\beta R(-b_{1}^{p^m-2}b_{2}))+T(Tr_{s}^{m}(\alpha R(-b_{1}^{p^m-2}b_{2}))).
\end{equation*}
\end{small}Since $H_{c}$ is weakly regular bent, we have that $G_{c}=(H_{c})^{*}$ is also weakly regular bent and $G$ is vectorial bent. Thus, $H$ is vectorial dual-bent with $(H_{c})^{*}=(H^{*})_{c}$ and $\varepsilon_{H_{c}}=\varepsilon$ for any $c \in \mathbb{F}_{p^s}^{*}$, where $H^{*}=G$, that is, $H$ satisfies Condition A. By Theorem 1, the partition $\Gamma$ generated from $H$ is a bent partition satisfying Condition $\mathcal{C}$.
\end{proof}

The following explicit construction of bent partitions is an immediate result of Proposition 3 and Theorem 4.

\begin{theorem}\label{5}
Let $n, m, s$ be positive integers with $s \mid n, s \mid m, s\neq m$, and $u_{i}, i \in \mathbb{F}_{p^s}$ be integers for which for any $i \in \mathbb{F}_{p^s}$, $u_{i}\equiv p^{j_{i}} \ mod \ (p^s-1)$ for some $0\leq j_{i}\leq s-1$ and $gcd(u_{i}, p^n-1)=1$. For any $i \in \mathbb{F}_{p^s}$, let $d_{i}$ be an integer with $u_{i}d_{i}\equiv 1 \ mod \ (p^n-1)$, and $P_{i}=(\mathbb{F}_{p^n}, +, \circ_{i})$ be a (pre)semifield for which its dual $P_{i}^{\star}$ is right $\mathbb{F}_{p^s}$-linear. For any $i \in \mathbb{F}_{p^s}$, let $F(i; x_{1}, x_{2}): \mathbb{F}_{p^n} \times \mathbb{F}_{p^n}\rightarrow \mathbb{F}_{p^s}$ be an arbitrary vectorial dual-bent function constructed by Proposition 3 with $u=u_{i}, d=d_{i}, P=P_{i}$. Let $\alpha, \beta \in \mathbb{F}_{p^m}$ be linearly independent over $\mathbb{F}_{p^s}$, $R$ be a permutation over $\mathbb{F}_{p^m}$ with $R(0)=0$ and $T: \mathbb{F}_{p^s}\rightarrow \mathbb{F}_{p^s}$ be an arbitrary function. Define $H: \mathbb{F}_{p^n} \times \mathbb{F}_{p^n} \times \mathbb{F}_{p^m} \times \mathbb{F}_{p^m}\rightarrow \mathbb{F}_{p^s}$ as
\begin{small}
\begin{equation*}
  H(x_{1}, x_{2}, y_{1}, y_{2})=F(Tr_{s}^{m}(\alpha R(y_{1}y_{2}^{p^m-2})); x_{1}, x_{2})+Tr_{s}^{m}(\beta R(y_{1}y_{2}^{p^m-2}))+T(Tr_{s}^{m}(\alpha R(y_{1}y_{2}^{p^m-2}))).
\end{equation*}
\end{small}Then
\begin{equation*}
  \Gamma=\{A_{i}, i \in \mathbb{F}_{p^s}\}
\end{equation*}
is a bent partition satisfying Condition $\mathcal{C}$, where
\begin{equation*}
  A_{i}=\{(x_{1}, x_{2}, y_{1}, y_{2})\in \mathbb{F}_{p^n} \times \mathbb{F}_{p^n} \times \mathbb{F}_{p^m} \times \mathbb{F}_{p^m}: H(x_{1}, x_{2}, y_{1}, y_{2})=i\}.
\end{equation*}
\end{theorem}

\begin{remark}\label{3}
With the same notation as in Theorem 4. Note that in Theorem 4, by setting vectorial dual-bent functions $H$ constructed by Theorem 5 as building blocks (that is, as $F(i; x)$), we can obtain more explicit vectorial dual-bent functions which can generate more bent partitions by Theorem 4.
\end{remark}

We give an example by using Theorem 5.

\begin{example}\label{2}
Let $p=3, s=4, n=m=8$. Let $\alpha$ be a primitive element of $\mathbb{F}_{3^8}$ and $\beta=1$, $R$ be the identity map and $T=0$. For any $i \in \mathbb{F}_{3^4}$, let
\begin{equation*}
  F(i; x_{1}, x_{2})=\left\{\begin{split}
                    Tr_{4}^{8}(x_{1}^{-89}x_{2}), & \text{ if } i \in \mathbb{F}_{3^4}^{*},\\
                    Tr_{4}^{8}(x_{1}x_{2}^{-83}),  & \text{ if } i=0.
                 \end{split}
  \right.
\end{equation*}
Then
\begin{equation*}
  H(x_{1}, x_{2}, y_{2}, y_{2})=(Tr_{4}^{8}(\alpha y_{1}y_{2}^{6559}))^{80}(Tr_{4}^{8}(x_{1}^{-89}x_{2}-x_{1}x_{2}^{-83}))+Tr_{4}^{8}(x_{1}x_{2}^{-83}+y_{1}y_{2}^{6559}),
\end{equation*}
and $\Gamma=\{D_{H, i}, i \in \mathbb{F}_{3^4}\}$ is a bent partition satisfying Condition $\mathcal{C}$, where $D_{H, i}=\{(x_{1}, x_{2}, y_{1}, y_{2}) \in (\mathbb{F}_{3^8})^{4}: H(x_{1}, x_{2}, y_{1}, y_{2})=i\}$.
\end{example}

\section{Relations between bent partitions and partial difference sets}
\label{sec:5}
In this section, by taking vectorial dual-bent functions as the link between bent partitions and partial difference sets, we give a sufficient condition on constructing partial difference sets from bent partitions. When $p$ is an odd prime, we characterize bent partitions satisfying Condition $\mathcal{C}$ in terms of partial difference sets.

\begin{definition}\label{3}
Let $(G, +)$ be a finite abelian group of order $v$ and $D$ be a subset of $G$ with $k$ elements. Then $D$ is called a \textit{$(v,k,\lambda,\mu)$ partial difference set} of $G$, if the expressions $d_{1}-d_{2}$, for $d_{1}$ and $d_{2}$ in $D$ with $d_{1}\neq d_{2}$, represent each nonzero element in $D$ exactly $\lambda$ times, and represent each nonzero element in $G \ \backslash \ D$ exactly $\mu$ times. When $\lambda=\mu$, then $D$ is called a \textit{$(v, k, \lambda)$ difference set}.
\end{definition}

Note that if $D$ is a partial difference set of $G$ with $-D=D$, then so are $D \cup \{0\}, D\ \backslash \ \{0\}, G \ \backslash \ D$ (see \cite{Ma1994A}). There is an important tool to characterize partial difference sets in terms of characters.

\begin{lemma}[\cite{Ma1994A}]\label{3}
Let $G$ be an abelian group of order $v$. Suppose that $D$ is a subset of $G$ with $k$ elements which satisfies $-D=D$ and $0 \notin D$. Then $D$ is a $(v, k, \lambda, \mu)$ partial difference set if and only if for each non-principal character $\chi$ of $G$,
\begin{equation*}
  \chi(D)=\frac{\beta\pm \sqrt{\Delta}}{2},
\end{equation*}
where $\chi(D)=\sum_{x \in D}\chi(x)$, $\beta=\lambda-\mu, \gamma=k-\mu, \Delta=\beta^{2}+4\gamma$.
\end{lemma}

When $p$ is an odd prime or $s\geq 2$, we give the value distribution of vectorial dual-bent functions satisfying Condition A.

\begin{proposition}\label{4}
Let $F: V_{n}^{(p)}\rightarrow V_{s}^{(p)}$ be a vectorial dual-bent function satisfying Condition A, where $p$ is odd or $s\geq 2$. Then
\begin{equation*}
  |D_{F, F(0)}|=p^{n-s}+\varepsilon p^{\frac{n}{2}-s} (p^s-1),  \ |D_{F, i}|=p^{n-s}-\varepsilon p^{\frac{n}{2}-s} \text{ if } i\neq F(0).
\end{equation*}
\end{proposition}

\begin{proof}
Note that if $f$ is a weakly regular $p$-ary bent function, then for any $a \in \mathbb{F}_{p}$, $f-a$ is a weakly regular bent function with $(f-a)^{*}=f^{*}-a$ and $\varepsilon_{f-a}=\varepsilon_{f}$. Since $F$ is a vectorial dual-bent function with $(F_{c})^{*}=(F^{*})_{c}, c \in V_{s}^{(p)}\backslash \{0\}$, we have that $F(x)-F(0)$ is a vectorial bent function and for any $c \in V_{s}^{(p)}\backslash \{0\}$,
\begin{equation*}
  ((F-F(0))_{c})^{*}=(F_{c})^{*}-\langle c, F(0)\rangle_{s}=(F^{*})_{c}-\langle c, F(0)\rangle_{s}=(F^{*}-F(0))_{c},
\end{equation*}
which implies that $F(x)-F(0)$ is a vectorial dual-bent function with $((F-F(0))_{c})^{*}=(F^{*}-F(0))_{c}$ and $\varepsilon_{(F-F(0))_{c}}=\varepsilon$ for any $c \in V_{s}^{(p)}\backslash \{0\}$. By the proof of Theorem 1, $F(ax)=F(x)$ for any $a \in \mathbb{F}_{p}^{*}$ and thus $F(x)=F(-x)$. By Corollary 1 of \cite{WF2022Ne} (Note that although Corollary 1 of \cite{WF2022Ne} only considers the case of $p$ being odd, the conclusion of Corollary 1 of \cite{WF2022Ne} also holds for $p=2, s\geq 2$), we have
\begin{equation*}
  |D_{F-F(0), 0}|=p^{n-s}+\varepsilon p^{\frac{n}{2}-s} (p^s-1),  \ |D_{F-F(0), i}|=p^{n-s}-\varepsilon p^{\frac{n}{2}-s} \text{ if } i\neq 0,
\end{equation*}
that is, \begin{equation*}
  |D_{F, F(0)}|=p^{n-s}+\varepsilon p^{\frac{n}{2}-s}(p^s-1),  \ |D_{F, i}|=p^{n-s}-\varepsilon p^{\frac{n}{2}-s} \text{ if } i\neq F(0).
\end{equation*}
\end{proof}

In the following, we give a characterization of vectorial dual-bent functions satisfying Condition A in terms of partial difference sets.

\begin{theorem}\label{6}
Let $n$ be an even positive integer, $s$ be a positive integer with $s\leq \frac{n}{2}$, and $F: V_{n}^{(p)}\rightarrow V_{s}^{(p)}$. The following two statements are equivalent.

(1) $F$ is a vectorial dual-bent function satisfying Condition A.

(2) When $p=2, s=1$, then the support $supp(F)$ of $F$ defined as $supp(F)=\{x \in V_{n}^{(2)}: F(x)=1\}$ is a $(2^{n}, 2^{n-1}\pm 2^{\frac{n}{2}-1}, 2^{n-2}\pm 2^{\frac{n}{2}-1})$ difference set, and when $p$ is odd or $s\geq 2$, then for any nonempty set $I\subseteq V_{s}^{(p)}$, $D_{F, I} \backslash \{0\}$ is a $(p^n, k, \lambda, \mu)$ partial difference set for which $-D_{F, I}=D_{F, I}$ and if $F(0) \in I$, then\\
\begin{equation}\label{35}
  \begin{split}
      & k=p^{n-s}|I|+\varepsilon p^{\frac{n}{2}-s}(p^s-|I|)-1,\\
      & \lambda=p^{n-2s}|I|^{2}+\varepsilon p^{\frac{n}{2}-s}(p^s-|I|)-2,\\
      & \mu=p^{n-2s}|I|^{2}+\varepsilon p^{\frac{n}{2}-s}|I|,
  \end{split}
\end{equation}
and if $F(0) \notin I$, then
\begin{equation}\label{36}
  \begin{split}
       & k=p^{n-s}|I|-\varepsilon p^{\frac{n}{2}-s}|I|, \\
       & \lambda=p^{n-2s}|I|^{2}+\varepsilon p^{\frac{n}{2}-s}(p^s-3|I|),\\
       & \mu=p^{n-2s}|I|^{2}-\varepsilon p^{\frac{n}{2}-s}|I|,
   \end{split}
\end{equation}
where $D_{F, I}=\{x \in V_{n}^{(p)}: F(x) \in I\}$ and $\varepsilon \in \{\pm 1\}$ is a constant ($\varepsilon=1$ if $p=2$).
\end{theorem}

\begin{proof}
It is easy to see that a Boolean function $F$ is a vectorial dual-bent function satisfying Condition A if and only if $F$ is bent, that is, Condition A is trivial for any Boolean bent function. By the well-known result that a Boolean function $F: V_{n}^{(2)}\rightarrow \mathbb{F}_{2}$ is bent if and only if its support $supp(F)=\{x \in V_{n}^{(2)}: F(x)=1\}$ is a $(2^{n}, 2^{n-1}\pm 2^{\frac{n}{2}-1}, 2^{n-2}\pm 2^{\frac{n}{2}-1})$ difference set (see \cite{Dillon1974El}), the conclusion obviously holds for $p=2, s=1$. In the following, we prove the conclusion for $p$ being odd or $s\geq 2$.

\noindent $(1)\Rightarrow (2)$: By the proof of Theorem 1, $F(-x)=F(x)$, that is, $-D_{F, I}=D_{F, I}$. For any $u \in V_{n}^{(p)} \backslash \{0\}$, with the same argument as in the proof of Theorem 2 of \cite{WF2022Ne},
\begin{equation*}
    \chi_{u}(D_{F, I}) =\left\{\begin{split}
                               \varepsilon p^{\frac{n}{2}}-\varepsilon p^{\frac{n}{2}-s}|I|, & \ \ \text{if} \ F^{*}(-u) \in I, \\
                               -\varepsilon p^{\frac{n}{2}-s}|I|, & \ \ \text{if} \ F^{*}(-u) \notin I.
                          \end{split}\right.
\end{equation*}
where $\varepsilon=1$ if $p=2$ since all Boolean bent functions are regular.

If $F(0) \in I$, then $|D_{F, I}\backslash \{0\}|=|D_{F, I}|-1$ and $\chi_{u}(D_{F, I}\backslash \{0\})=\chi_{u}(D_{F, I})-1$. By Proposition 4, $|D_{F, I} \backslash \{0\}|=(|I|-1)(p^{n-s}-\varepsilon p^{\frac{n}{2}-s})+(p^{n-s}+\varepsilon p^{\frac{n}{2}-s}(p^s-1)-1)=p^{n-s}|I|+\varepsilon p^{\frac{n}{2}-s}(p^s-|I|)-1$. By Lemma 3, $D_{F, I}\backslash \{0\}$ is a $(p^n, k, \lambda, \mu)$ partial difference set, where $k, \lambda, \mu$ are given in (35).

If $F(0) \notin I$, then $|D_{F, I}\backslash \{0\}|=|D_{F, I}|$ and $\chi_{u}(D_{F, I}\backslash \{0\})=\chi_{u}(D_{F, I})$. By Proposition 4, $|D_{F, I} \backslash \{0\}|=|I|(p^{n-s}-\varepsilon p^{\frac{n}{2}-s})$. By Lemma 3, $D_{F, I}\backslash \{0\}$ is a $(p^n, k, \lambda, \mu)$ partial difference set, where $k, \lambda, \mu$ are given in (36).

\noindent $(2)\Rightarrow (1)$: By Lemma 3, for any $u \in V_{n}^{(p)}$ and nonempty set $I \subseteq V_{s}^{(p)}$ we have
\begin{equation}\label{37}
  \chi_{u}(D_{F, I})=p^{n-s}\delta_{\{0\}}(u)|I|+\varepsilon p^{\frac{n}{2}}-\varepsilon p^{\frac{n}{2}-s}|I| \text{ or }\chi_{u}(D_{F, I})=p^{n-s}\delta_{\{0\}}(u)|I|-\varepsilon p^{\frac{n}{2}-s}|I|.
\end{equation}
For any $i \in V_{s}^{(p)}$, define $W_{i}=\{u \in V_{n}^{(p)}: \chi_{u}(D_{F, i})=p^{n-s}\delta_{\{0\}}(u)+\varepsilon p^{\frac{n}{2}}-\varepsilon p^{\frac{n}{2}-s}\}$, where $D_{F, i}=\{x \in V_{n}^{(p)}: F(x)=i\}$. We claim that $W_{i}\bigcap W_{i'}=\emptyset$ for any $i \neq i'$ and $\bigcup_{i \in V_{s}^{(p)}}W_{i}=V_{n}^{(p)}$. Indeed, if there exist $i\neq i'$ such that $W_{i}\bigcap W_{i'}\neq \emptyset$, that is, there exists $u \in V_{n}^{(p)}$ such that $\chi_{u}(D_{F, i})=\chi_{u}(D_{F, i'})=p^{n-s}\delta_{\{0\}}(u)+\varepsilon p^{\frac{n}{2}}-\varepsilon p^{\frac{n}{2}-s}$, then $\chi_{u}(D_{F, i}\bigcup D_{F, i'})=2p^{n-s}\delta_{\{0\}}(u)+2\varepsilon p^{\frac{n}{2}}-2\varepsilon p^{\frac{n}{2}-s}$, which contradicts with (37). Thus, $W_{i}\bigcap W_{i'}=\emptyset$ for any $i\neq i'$. If there exists $u \in V_{n}^{(p)}$ such that $u \notin W_{i}$ for any $i \in V_{s}^{(p)}$, that is, $\chi_{u}(D_{F, i})=p^{n-s}\delta_{\{0\}}(u)-\varepsilon p^{\frac{n}{2}-s}$ for any $i \in V_{s}^{(p)}$, then $\chi_{u}(V_{n}^{(p)})=\sum_{i \in V_{s}^{(p)}}\chi_{u}(D_{F, i})=p^{n}\delta_{\{0\}}(u)-\varepsilon p^{\frac{n}{2}}$, which contradicts with $\chi_{u}(V_{n}^{(p)})=\sum_{x \in V_{n}^{(p)}}\zeta_{p}^{\langle u, x\rangle_{n}}=p^{n}\delta_{\{0\}}(u)$. Thus, $\bigcup_{i \in V_{s}^{(p)}}W_{i}=V_{n}^{(p)}$. By the above analysis, we have
\begin{equation}\label{38}
  \chi_{u}(D_{F, I})=p^{n-s}\delta_{\{0\}}(u)|I|+\varepsilon p^{\frac{n}{2}-s}(p^s\delta_{W_{I}}(u)-|I|),
\end{equation}
where $W_{I}=\sum_{i \in I}W_{i}$. By (38) and Lemma 1, $F$ is a vectorial dual-bent function satisfying Condition A.
\end{proof}

The following theorem provides a sufficient condition on constructing partial difference sets from bent partitions.

\begin{theorem}\label{7}
Let $n$ be an even positive integer and $s$ be a positive integer with $s\leq \frac{n}{2}$. Assume that $\Gamma=\{A_{i}, i \in V_{s}^{(p)}\}$ is a bent partition of $V_{n}^{(p)}$ for which the function $F: V_{n}^{(p)}\rightarrow V_{s}^{(p)}$ defined by
\begin{equation*}
  F(x)=i \text{ if } x \in A_{i}
\end{equation*}
is a vectorial dual-bent function satisfying Condition A. Then when $p=2, s=1$, $A_{0}$ and $A_{1}$ are $(2^{n}, 2^{n-1}\pm 2^{\frac{n}{2}-1}, 2^{n-2}\pm 2^{\frac{n}{2}-1})$ difference set and $(2^{n}, 2^{n-1}\mp 2^{\frac{n}{2}-1}, 2^{n-2}\mp 2^{\frac{n}{2}-1})$ difference set, respectively,
and when $p$ is odd or $s\geq 2$, for any nonempty set $I\subseteq V_{s}^{(p)}$, $A_{I} \backslash \{0\}=\bigcup_{i \in I}A_{i} \backslash \{0\}$ is a $(p^n, k, \lambda, \mu)$ partial difference set, where $(k, \lambda, \mu)$ are given in (35) if $0 \in A_{I}$ and $(k, \lambda, \mu)$ are given in (36) if $0 \notin A_{I}$.
\end{theorem}

\begin{proof}
Note that if $D$ is a $(v, k, \lambda)$ difference set of a finite abelian group $G$, then $G\backslash D$ is a $(v, v-k, v-2k+\lambda)$ difference set of $G$ (for instance see \cite{Ding2015Co}). Then the result follows from Theorem 6.
\end{proof}

\begin{remark}\label{4}
By Proposition 3, the bent partition $\Gamma_{1}$ (resp. $\Gamma_{2}$, $\Gamma_{1}^{\bullet}$, $\Gamma_{2}^{\bullet}$, $\Theta_{1}$, $\Theta_{2}$) satisfies the condition in Theorem 7. By Theorem 7, any union of sets from $\Gamma_{1}$ (resp, $\Gamma_{2}$, $\Gamma_{1}^{\bullet}$, $\Gamma_{2}^{\bullet}$, $\Theta_{1}$, $\Theta_{2}$) forms a partial difference set. Thus, the results given in Corollary 15 of \cite{AK2022A} on constructing partial difference sets from $\Gamma_{1}$ (resp. $\Gamma_{2}$, $\Gamma_{1}^{\bullet}$, $\Gamma_{2}^{\bullet}$, $\Theta_{1}$, $\Theta_{2}$) (which includes the results given in Theorem 2 of \cite{AKM2022Be} on constructing partial difference sets from $\Gamma_{1}$, resp. $\Gamma_{2}$, in the finite field) can also be illustrated by our results.
\end{remark}

Since the bent partitions constructed in Theorem 5 satisfy the condition in Theorem 7, we have the following corollary from Theorem 7.

\begin{corollary}\label{2}
Let $\Gamma=\{A_{i}, i \in \mathbb{F}_{p^s}\}$ be a bent partition constructed by Theorem 5. Then when $p=2, s=1$, $A_{0}$ and $A_{1}$ are $(2^{n}, 2^{n-1}\pm 2^{\frac{n}{2}-1}, 2^{n-2}\pm 2^{\frac{n}{2}-1})$ difference set and $(2^{n}, 2^{n-1}\mp 2^{\frac{n}{2}-1}, 2^{n-2}\mp 2^{\frac{n}{2}-1})$ difference set, respectively, and when $p$ is odd or $s\geq 2$, for any nonempty set $I\subseteq \mathbb{F}_{p^s}$, $A_{I} \backslash \{0\}=\bigcup_{i \in I}A_{i} \backslash \{0\}$ is a $(p^n, k, \lambda, \mu)$ partial difference set, where $(k, \lambda, \mu)$ are given in (35) with $\varepsilon=1$ if $0 \in A_{I}$ and $(k, \lambda, \mu)$ are given in (36) with $\varepsilon=1$ if $0 \notin A_{I}$.
\end{corollary}

We give an example by Corollary 2.

\begin{example}\label{3}
Let $\Gamma=\{D_{H, i}, i\in \mathbb{F}_{3^4}\}$ be the bent partition constructed in Example 2. By Corollary 2, $D_{H, i}$ is a
\begin{small}$(1853020188851841, 22876791923520, 282470988879, 282429005040)$ \end{small}partial difference set for any $i \in \mathbb{F}_{3^4}^{*}$, $D_{H, 0}\backslash \{0\}$ is a \begin{small}$(1853020188851841, 22876834970240, 282472051759, 282430067922)$\end{small} partial difference set, $(D_{H, 0}\bigcup D_{H, 1})\backslash \{0\}$ is a \begin{small}$(1853020188851841, 45753626893760, 1129760129761, $ $1129719208806)$\end{small} partial difference set.
\end{example}

When $p$ is an odd prime, we immediately obtain the following characterization of bent partitions of $V_{n}^{(p)}$ satisfying Condition $\mathcal{C}$ from Theorems 3 and 6.

\begin{theorem}\label{8}
Let $p$ be an odd prime. Let $\Gamma=\{A_{i}, i \in V_{s}^{(p)}\}$ be a partition of $V_{n}^{(p)}$, where $n$ is even and $s\leq \frac{n}{2}$. Then the following two statements are equivalent.

(1) $\Gamma$ is a bent partition satisfying Condition $\mathcal{C}$.

(2)  For any nonempty set $I\subseteq V_{s}^{(p)}$, $A_{I} \backslash \{0\}=\bigcup_{i \in I}A_{i} \backslash \{0\}$ is a $(p^n, k, \lambda, \mu)$ partial difference set with $-A_{I}=A_{I}$, where $(k, \lambda, \mu)$ are given in (35) if $0 \in A_{I}$ and $(k, \lambda, \mu)$ are given in (36) if $0 \notin A_{I}$.
\end{theorem}

\section{Conclusion}
\label{sec:6}
In this paper, we investigated relations between bent partitions and vectorial dual-bent functions (Theorems 1, 2, 3) and gave some new constructions of bent partitions satisfying Condition $\mathcal{C}$ (Corollary 1, Theorems 4 and 5). We illustrated that for any ternary weakly regular bent function $f: V_{n}^{(3)}\rightarrow \mathbb{F}_{3}$ ($n$ even) with $f(x)=f(-x)$ and $\varepsilon_{f}=-1$, the generated bent partition by $f$ is not coming from a normal bent partition (see Example 1), which answers an open problem proposed in \cite{AM2022Be}. By taking vectorial dual-bent functions as the link between bent partitions and partial difference sets, we give a sufficient condition on constructing partial difference sets from bent partitions (Theorem 7). When $p$ is an odd prime, we characterized bent partitions satisfying Condition $\mathcal{C}$ in terms of partial difference sets (Theorem 8).

\end{document}